\documentclass[twocolumn]{svjour3}          
\smartqed  

\usepackage{graphicx}

\usepackage{float}
\usepackage{caption} 
\usepackage{subcaption}
\usepackage{color}
\usepackage{graphicx}
\usepackage{epsfig}
\usepackage{url}
\usepackage{version}
\usepackage{enumitem}
\usepackage{mdframed}
\usepackage{amsmath}
\DeclareCaptionFont{tiny}{\tiny}
\usepackage{colortbl}
\usepackage{amssymb}
\usepackage{tcolorbox}
\usepackage{etoolbox}  
\AtBeginEnvironment{tcolorbox}{\normalsize} 
\usepackage{xspace}
\usepackage{tablefootnote}
\usepackage{longtable}

\definecolor{light-gray}{gray}{0.96}
\definecolor{LightCyan}{rgb}{0.88,1,1}

\newcommand{\defn}[1]{\textbf{\emph{#1}}}

\definecolor{darkblue}{rgb}{0.0,0.0,0.75}
\definecolor{ultramarine}{rgb}{0.07, 0.04, 0.76}

\newcommand{\B}{\textsc{BEB}\xspace}
\newcommand{\LB}{\textsc{Log-Backoff}\xspace}
\newcommand{\LLB}{\textsc{LogLog-Backoff}\xspace}
\newcommand{\STB}{\textsc{Sawtooth-Backoff}\xspace}

\newcommand{\totaltime}{execution time\xspace}
\newcommand{\Totaltime}{Execution time\xspace}

\newcommand{\detailedSimulator}{detailed simulator\xspace}

\newcommand{\theoreticalSimulator}{theoretical-model simulator\xspace}

\newtheorem{fact}{Fact}

\newcolumntype{M}[1]{>{\centering\arraybackslash}m{#1}} 

\begin{document}

\vspace{-2cm}

\title{Windowed Backoff Algorithms for WiFi \thanks{This research is supported by the National Science Foundation grant CNS-1816076 and the U.S. National Institute of Justice (NIJ) Grant 2018-75-CX-K002.}
}
\subtitle{Theory and Performance under Batched Arrivals}


\author{\mbox{William C. Anderton $\cdot$ Trisha Chakraborty $\cdot$ Maxwell Young}}


\institute{ \at W. C. Anderton\\
              \email{wca36@msstate.edu}   
              \and     
           T. Chakraborty \at
              Department of Computer Science and Engineering, Mississippi State University, MS, USA\\
               \email{tc2006@msstate.edu}
           \and
               M. Young (contact author)\at  Department of Computer Science and Engineering, Mississippi State University, MS, USA\\
               \email{myoung@cse.msstate.edu}
}

\date{}

\maketitle

\begin{abstract}
\hspace{-3pt}Binary exponential backoff (BEB) is a decades -old algorithm for coordinating access to a shared channel. In modern networks, BEB plays a crucial role in WiFi  and other wireless communication standards.

Despite this track record, well-known theoretical results indicate that under bursty traffic, BEB yields poor makespan, and superior algorithms are possible. To date, the degree to which these findings impact performance in wireless networks has not been examined.

Here, we investigate a challenging case for BEB: a single burst (\defn{batch}) of packets that simultaneously contend for access to a wireless channel. Using Network Simulator 3, we incorporate into IEEE 802.11g several newer algorithms that have  theoretically-superior makespan guarantees. Surprisingly, we discover that these newer algorithms underperform BEB.

Investigating further, we identify as the culprit a common abstraction regarding the cost of collisions.  Our experimental results are complemented by analytical arguments that the number of collisions---and not solely makespan---is an important metric to optimize. We propose a new theoretical model that accounts for the cost of collisions, and derive new asymptotic bounds on the makespan for BEB and the newer backoff algorithms that align with our experimental findings. Finally, we argue that these findings have implications for the design of backoff algorithms in wireless networks.\smallskip
\end{abstract}

\section{Introduction}\label{sec:intro}\vspace{-5pt}

\defn{Randomized binary exponential backoff} (\defn{BEB}) plays an important role in coordinating access by multiple devices to a shared resource. Originally designed for use in old Ethernet systems decades ago~\cite{MetcalfeBo76}, BEB has since found application in a range of domains such as transactional memory~\cite{herelihy:transactional,scherer:advanced}, concurrent memory access~\cite{Ben-DavidB17,80120,mellor-crummey:algorithms}, and congestion control~\cite{mondal:removing}.  However, arguably, the most prominent application of BEB today is in IEEE 802.11 (WiFi) networks, where the shared resource is a wireless communication channel. 

Given its importance, BEB has been studied at length and is known to yield good throughput under well-behaved traffic~\cite{GoldbergMa96a,GoldbergMaPaSr00,HastadLeRo87,RaghavanUp95,Al-Ammal2000,Al-Ammal2001,Goodman:1988:SBE:44483.44488,bianchi:performance,song:stability}.  In contrast, when traffic is bursty, BEB is suspected to perform sub-optimally.  This is true, even for \defn{batched arrivals}, where  {\boldmath{$n$}} \defn{packets}  simultaneously arrive in a time-slotted system, begin contending for the channel, and succeed in being transmitted before the next batch arrives. In this setting, Bender~et al.~\cite{BenderFaHe05} prove that BEB has  $\Theta(n\log n)$ \defn{makespan}, which is the number of time slots until all packets are successfully transmitted.  In other words, the throughput is $O(1/\log n)$ and thus tends asymptotically to zero.

In light of this shortcoming, there has been significant interest in developing algorithms with improved makespan. Under batched arrivals, Bender et al.~\cite{BenderFaHe05} derive upper and lower bounds for several variants of backoff that are asymptotically superior to BEB in this regard. Furthermore,  again under batched arrivals, sawtooth backoff~\cite{GreenbergFlLa87} and truncated sawtooth backoff~\cite{bender:contention} algorithms achieve the asymptotically optimal makespan of $O(n)$.

This growing body of results provokes an obvious question: {\it How do newer algorithms compare to BEB in practice?} Here, we make progress towards an answer by restricting ourselves to bursty wireless traffic. In particular, we examine the case of a single burst (\defn{batch}) of packets. This is a prominent case in the theoretical literature where BEB is anticipated to do poorly, and it should be possible to identify which of the following situations is true: (1) A newer contention-resolution algorithm outperforms BEB, or (2) BEB outperforms newer contention-resolution algorithms.

Interestingly, neither of these outcomes is very palatable. In one form or another, BEB has operated in networks for over four decades and it remains an essential ingredient in several wireless standards.  Bursty traffic can arise in practice~\cite{Teymori2005,yu:study} and its impact has been examined~\cite{Ghani:2010,sarkar:effect,canberk:self,bhandari:performance}. If (1) holds, then BEB is potentially in need of revision and the ramifications of this are hard to overstate.  

Conversely, if (2) holds, then theory is not translating into improved performance in a prominent application domain.  At best, this is a matter of asymptotics. At worst, this indicates a problem with the theoretical model upon which newer results  are based. In this latter case, it is important to understand what assumptions are faulty, so that the model may be revised. \vspace{-10pt}


\subsection{A Common Theoretical Model} \label{sec:common-model}

In a wireless setting, given $n$ stations, the problem of \defn{contention resolution} addresses the number of slots until any one of the  stations transmits alone. A natural consideration is the time until a subset of $k$ stations each transmits alone; this often falls under the same label, but is also sometimes referred to as {\boldmath{$k$}}\defn{-selection} in the literature~(for example, see~\cite{Anta2010}). 

Here, we focus on the case of $k=n$ and examine the performance of various backoff algorithms. Much of the algorithmic work shares a theoretical model.  Three common assumptions are:\smallskip

\begin{itemize}[leftmargin=3.5mm]
\item{\bf A0.} Each slot has length that can accommodate a packet. \medskip

\item {\bf A1.} If a single packet transmits in a slot, the packet \defn{succeeds}, but \defn{failure} occurs if two or more packets transmit simultaneously due to a \defn{collision}.  \medskip

\item {\bf A2.} A collision incurs a delay of a single slot in which the failure occurred.
\end{itemize}

Assumption A0 is near-universal, but technically inaccurate for reasons discussed in Section~\ref{sec:802.11}. To summarize, under BEB, each station selects a random slot from a set of $2^r$ consecutive slots, for some integer $r\geq 0$; this is equivalent to setting a counter randomly to a value in
 $\{0, 1, ..., 2^{r}-1\}$, which is decremented at each slot, and transmission occurs when the counter reaches $0$. Thus, from an algorithmic perspective, stations decrement their respective counters as they march through these slots uninterrupted. However, in practice, transmission of the full packet may occur past the slot, while all other stations pause their execution (not decrementing their respective counters) until the transmission ends. Although desynchronization can occur, this assumption seems sufficiently close to reality that we should not expect performance to deviate greatly as a result.

Assumption A1 is prevalent in the literature (see~\cite{Anta2010,komlos:asymptotically,Capetanakis:2006,bender:heterogeneous,bender:contention,bender:how,fineman:contention2}), although variations exist. A  compelling alternative is the signal-to-noise-plus-interference (SINR) model~\cite{avin:sinr,moscibroda:worst},  which is less strict about failure in the event of simultaneous transmissions. Another model that has received attention is the affectance model~\cite{Hall2009}. Nevertheless, these all share the assumption that simultaneous transmissions may lower the probability of success.

Assumption A2 is also widely adopted (see the same examples for A1) and implicitly addresses two quantities that affect performance: the time to transmit a packet, and the time to receive any channel feedback on success or failure. Assigning a negligible delay to these quantities admits a simplified model, but ignores the associated performance impact. For example, the functionality for obtaining channel feedback is provided by a medium access control (MAC) protocol, of which a backoff algorithm is only one component.\vspace{-10pt}


\subsection{Our Main Message}

Our main message is that A2 is flawed in the WiFi setting. In particular, the cost of a collision is more significant than acknowledged by the theoretical model. This is not a matter of minor adjustments to the assumption, or an artifact of hidden constants in the algorithms examined. Rather, in WiFi networks, the way in which collisions are detected requires a revision to the problem of contention resolution if we aim to design algorithms for many wireless settings.

Several corollaries follow from our main claim, which we illustrate using detailed simulations. Additionally, we provide analytical arguments that support our experimental findings (see Section~\ref{sec:theory}).  Our belief is that these findings generalize to other  wireless settings---not only WiFi networks---and that contention-resolution algorithms that ignore the cost of collisions will likely not perform as advertised in these cases  (see Section~\ref{sec:interpret}). 

Finally, we emphasize that our findings are aimed at complementing prior theoretical results whose importance to many application domains is evident (see Section~\ref{sec:related}). However, arguably, a major application of backoff algorithms {\it is} in wireless networks, and this is what motivates our investigation.\vspace{-10pt}


\section{Overview of BEB in IEEE 802.11}\label{sec:802.11}

To understand our findings, it is helpful to summarize how many of the IEEE 802.11 standards work with  BEB. However, outside of this section and the description of our experimental setup,  discussion of such aspects and terminology is kept to a minimum. 

Throughout, we will often use interchangeably the terms {\it packet} and {\it station} depending on the context; the two uses are equivalent given that in the batched setting, each station seeks to transmit a single packet. Also, to be explicit, while we use the general term {\it packet}, we are specifically referring to a {\it frame}.

Exponential backoff~\cite{MetcalfeBo76} is a widely deployed algorithm for distributed multiple access. Informally, this backoff algorithm operates over a \defn{contention window (CW)} wherein each \defn{station} makes a single randomly-timed access attempt. In the event of two or more simultaneous attempts, the result is a collision and none of the stations succeed.  Backoff seeks to avoid collisions by dynamically increasing the contention-window size such that stations succeed.

IEEE 802.11 handles contention resolution via the \defn{distributed coordination function (DCF)} which employs BEB; as the name suggests, successive CWs double in size under BEB. The operation of DCF is summarized as follows. Prior to transmitting data, a station first senses the channel for a period of time known as a \defn{distributed inter-frame space (DIFS)}. If the channel is not in use over the DIFS, the station transmits its data; otherwise, it waits until the current transmission finishes and then initiates BEB.  

For a  contention window of size $w$, a timer value is selected uniformly at random from  $[0, w-1]$. So long as the channel is sensed to be idle, the timer counts down per slot and, when it expires, the station transmits. However, if at any prior time  the channel is sensed busy, BEB is {\it paused for the duration of the current transmission}, and then resumed (not restarted) after another DIFS.

After a station transmits, it awaits an {\bf acknowledgement (ACK)} from the receiver. If the transmission was successful, then the receiver waits for a short amount of time known as a \defn{short inter-frame space (SIFS)}---of shorter duration than a DIFS---before sending the ACK. Upon receiving an ACK, the station learns that its transmission was successful.  Otherwise, the station waits for an {\defn{ACK-timeout}} duration before concluding that a collision occurred.  This series of actions is referred to as \defn{collision detection}; the cost of which lies at the heart of our argument.  If a collision is detected, then the station must attempt a retransmission via the same process with its CW doubled in size.

Note that both the transmission of data and the acknowledgement process occur ``outside'' of the backoff component of DCF; Figure~1 highlights this. 
In contrast, the focus of many algorithmic results is solely on the slots of this backoff component. 

Finally,  RTS/CTS (request-to-send/clear-to-send) is an optional mechanism. Informally, a station will send an RTS message and await a CTS message from the receiver prior to transmitting its data. Due to increased overhead, there is some debate on whether RTS/CTS is worthwhile and, if so, when it should be enabled~\cite{chatzimisios2004effectiveness,xu2003effectiveness,chatzimisios2004optimisation}; typically, it is not. Here, we focus on the case where RTS/CTS is disabled, although our experiments show that our findings continue to hold when this mechanism is used (see Section~\ref{sec:hidden}). \vspace{-10pt}


\begin{figure}[t]
\vspace{-5pt}
\captionsetup[subfigure]{labelformat=empty}
\centering
\begin{subfigure}{1.0\textwidth} 
\hspace{0pt}\includegraphics[width=0.5\textwidth]{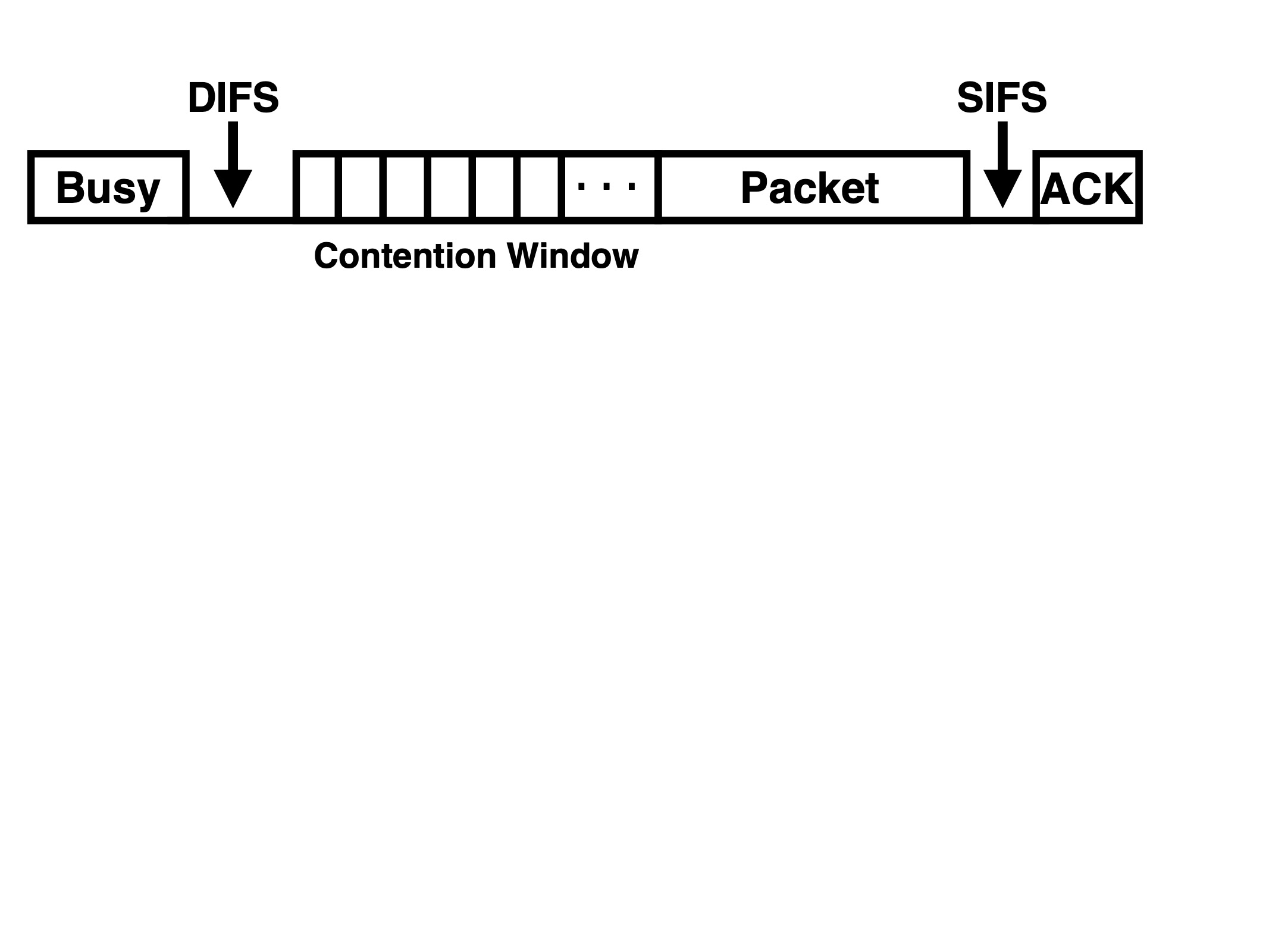} 
\vspace{-0.2cm}
\end{subfigure}
\vspace{-4.5cm}
\caption{An overview of DCF.} \vspace{-10pt}
\end{figure}


\section{Experimental Setup}\label{sec:experimental}

In this section, we describe and discuss our simulation tools, our choices about which protocols to simulate,  the network layout used in our experiments, and the method by which we exclude outliers from our data.\medskip

\noindent{\bf Simulation Tools.} We employ Network Simulator 3 (NS3, version 3.25)~\cite{NS3} which is a widely used network simulation tool in the research community~\cite{weingartner:performance}. Our simulation code, scripts, and data, and a description of these files are available online.\footnote{This content can be accessed by visiting \texttt{www.maxwellyoung.net/publications}, or by using Github at \texttt{https://github.com/trishac97/NS3-802.11g-Backoff}.} In order to motivate our design choices and for the purposes of reproducibility, our experimental setup is described in this section. 

Our reasons for using NS3 are twofold. First, wireless communication is difficult to model and employing NS3 helps allay concerns that our findings are an artifact of poorly-modeled wireless effects.  Second, given the assumptions upon which contention-resolution algorithms are based, NS3 can reveal whether we are being led astray by an assumption that appears reasonable, but results in an important discrepancy between theory and practice.

Table~\ref{table:parameters} provides the experimental parameters we use, unless we note otherwise. Path-loss models with default parameters are known to be faithful~\cite{stoffers:comparing} and, therefore, our experiments employ the  log-distance propagation loss model in NS3. For transmission and reception of packets, we use the YANS~\cite{lacage:yans} module. 

NS3 is valuable for the extraordinary level of detail it provides, and we will refer to this simulation tool as the {\defn{\detailedSimulator}}. However, this tool cannot scale to a large number of stations/packets. In cases where we wish to verify asymptotic results, we resort to a simpler simulation tool that implements the backoff algorithms using a notion of slots and collisions that aligns with the theoretical model (recall Section~\ref{sec:common-model}). We have created this in Java, and we refer to this simulation tool as the \defn{\theoreticalSimulator}.

 
\begin{table}[t] 
\begin{center}
{
\begin{tabular}{ |p{4.1cm}|p{3cm}|  }
\hline
\rowcolor{LightCyan}\hspace{37pt}{\bf Parameter} &  \hspace{30pt}{\bf Value}  \\
\hline
Slot duration & 9$\mu$s  \\
\hline
SIFS  &  16$\mu$s   \\
\hline
DIFS  &   34$\mu$s \\
\hline
ACK timeout &  75$\mu$s\\
\hline
Transport layer protocol & UDP\\
\hline
Packet overhead & 64 bytes\\
\hline
Contention-window size min. & 4\\
\hline
Contention-window size max.& 4096  \\
\hline
RTS/CTS & Off\tablefootnote{While our results focus on the more common case of RTS/CTS being disabled, we do briefly report on the impact of RTS/CTS in Section~\ref{sec:collisions-discussion}.} \\
\hline
\end{tabular}
}\caption{Parameter values used in our experiments, unless explicitly noted otherwise.}\label{table:parameters}
\end{center}
\end{table}


\medskip

\noindent{\bf Protocol Choices.} At the MAC layer, we make use of IEEE 802.11g in our NS3 experiments. IEEE 802.11g provides a data rate in the tens of megabits per second (Mbit/s)\footnote{The standard offers a theoretical maximum of 54 Mbit/s.}, operates in the 2.4 GHz band, and remains in use today. Although IEEE 802.11ac and IEEE 802.11n are more recent members of the WiFi family, we feel that IEEE 802.11g is sufficiently representative to establish our main thesis.

We implement changes to the behavior of the contention window based on the algorithms we investigate. All experiments use IPv4 and UDP. Our investigation employs UDP instead of TCP to reduce the impact of potential transport-layer effects that may complicate the interpretation of our results. Ultimately, given the explanation for our findings, we believe that this choice does not alter our final conclusions.

The amount of overhead for each packet is 64-bytes: 8 bytes for UDP, 20 bytes for IPv4, 8 bytes for an LLC/SNAP header, and 28 bytes of additional overhead at the MAC layer. Fragmentation is disabled in our experiments.

For our experiments, unless noted otherwise, we set our slot time to be  $9 \mu s$  and our a SIFS to be $16\mu s$.  The duration of an acknowledgement (ACK) timeout is specified by the  IEEE 802.11 standard\footnote{This timeout period is specified in  Section 10.3.2.9, page 1317 of~\cite{802.11-standard}.}  to be roughly the sum of a SIFS ($16\mu$s), standard slot time ($9\mu$s), and preamble ($20\mu$s); a total of $45\mu$s.  However, in practice, this is subject to tuning.  In our experiments, this ACK-timeout value gave poor performance; there is insufficient time for the ACK before the sender decides to retransmit. We set the default value of $75\mu$s in NS3 since this yielded good performance and is still the same order of magnitude as suggested in the standard.\medskip

\noindent{\bf Network Layout.} In our experiments, unless noted otherwise, $n$ stations are placed in a  $40$ meter $\times$ $40$ meter grid, and they are laid out at regular spacing, starting at the south-west corner of the grid moving left to right by $2$-meter (m) increments, and then up when the current row is filled.\footnote{We have also run experiments with $1$ and $3$ meter increments and the qualitative behavior remains the same; therefore, we omit those results.}  A \defn{wireless access point (AP)} is located (roughly) at the center of the grid. 


We do not simulate additional terrain or environmental phenomena. We aim for  a topology that (1) allows us to focus on the performance at the MAC-layer, and (2) approximates some situations in the real world. 

Regarding (1), the density of devices is uniform, so there is no subset of devices that will be subjected to a disproportionate amount of traffic and the   interference that may result. The central placement of the AP minimizes the maximum distance to all stations, which reduces signal attenuation.  Deviating from either of these two properties might degrade performance due to a positioning that exacerbates effects at the physical layer, and it could be challenging to separate such effects from the performance at the MAC-layer.

In the context of (2), we are more tentative. The range of WiFi under $2.4$GHz is often considered to be roughly $50$m  for indoor settings and $100$ meters for outdoor settings. This guides our choice of spacing; that is, no station is more than $40$m from our centrally-placed AP. For an indoor setting, there is obviously variation in how users will be located, but a grid might be viewed as a crude approximation to some scenarios, such as a large classroom setting. We note that the issue of AP placement is a research problem by itself (for example, see~\cite{zhong:methods,6503906,maksuriwong:moga}), with  specialized software tools~\cite{ekahau,aerohive} pertaining to just this aspect. However, a centrally placed AP might be viewed as an approximation to what is done in, for example, a lecture hall or a cafe.
\smallskip

\noindent{\bf Outliers.} Throughout, the following common approach is used to identify outliers in our data. Let $\Delta$ be the distance between the first and third quartiles, $Q_1$ and $Q_3$, respectively. Any data point smaller than $Q_1 - 1.5 \Delta$ or larger than $Q_3 + 1.5 \Delta$ is declared an outlier and discarded. We note that we observe very few outliers. 




\begin{figure}[t!]
\begin{tcolorbox}[standard jigsaw, opacityback=0]

\noindent{}\hspace{-3pt}{\bf Monotonic Windowed Backoff Algorithm}\smallskip

\noindent\hspace{-3pt}$w \leftarrow$ initial window size

\noindent\hspace{-3pt}$f(w)\leftarrow$ window-scaling function\medskip

\noindent{}\hspace{-3pt}A station with a packet to transmit does the following until successful:\vspace{-3pt}	 
	
          \begin{itemize}[leftmargin=3mm]
			
                 \item Attempt to transmit in a slot chosen uniformly at random from $w$.\smallskip
			
                  \item If the transmission failed, then wait until the end of the window and set $w \leftarrow  \lceil(1+ f(w))w\rceil$.
       
          \end{itemize}
\end{tcolorbox}
\caption{A generic backoff algorithm with monotically-increasing windows. In particular, the initial window size is a small constant, and $f(w)=1/\lg\lg w, 1/\lg w$, and $1$ for LLB,~LB, and \B, respectively.}\label{fig:pseusocode1}
 \end{figure}


\section{A Single Batch}\label{sec:single-batch}

We examine a single batch of $n$ packets that simultaneously begin their contention for the channel. As competitors with \B, we experiment with the following backoff algorithms: \LB ~(LB), \LLB (LLB)~\cite{BenderFaHe05} and \STB (STB) \cite{Gereb-GrausT92,GreenbergL85}. We also investigate the performance of \textsc{Truncated Sawtooth-Backoff} (TSTB)~\cite{bender:contention}, although we defer this to Section~\ref{sec:tstb}.

Both LLB and LB are closely related to \B~in that they execute using a CW that increases in size monotonically. Here, each window increases by a $(1+f(w))$-factor in size over the previous window, where $f(w)$ is a function that controls the amount by which the window size grows. For  LLB, LB, and BEB, $f(w)=1/\lg\lg w, 1/\lg w,$ and $1$, respectively, and the initial window sizes are $4$ when we execute our experiments. The pseudocode is presented in Figure~\ref{fig:pseusocode1}. 

In contrast, STB is non-monotonic and executes over a doubly-nested loop. The outer loop sets the current window size $w$ to be double that used in the preceding outer loop; this is like BEB. Additionally, for each such window, the inner loop executes over $O(\lg w)$ windows of decreasing size: $w, w/2, w/4, ..., 4$. In each window, a slot is chosen uniformly at random for packet transmission; this is the ``backon'' component of STB. 

Finally, TSTB executes identically to STB except that, for some constant $c>0$, the inner loop  executes over $\lg(c\lg w)$ windows of decreasing size: $w, w/2, w/4, ...,$  $\max\{\lfloor w/c\lg w \rfloor, 4\}$.


\medskip

\noindent{\bf Measuring Time.}  For experiments, the amount of time until all packets succeed is reported in microseconds ($\mu$s), and we refer to this as the \defn{\totaltime}.

The IEEE 802.11 standards also use a notion of a slot. For example, the \defn{contention window (CW)} used in the backoff component  consists of slots. The duration of a slot depends on the IEEE 802.11 standard employed, and we typically use a slot duration of $9 \mu s$.  We emphasize that other operations happen outside of this slotted time (recall Section~\ref{sec:802.11}), such as packet  transmission; all of this is captured by the \totaltime.\footnote{We exclude the time required for a station to associate with the access point and ARP requests/replies, as this will be the same regardless of the algorithm being evaluated.}

\begin{table}[t!] 
\begin{center}
{
\begin{tabular}{ |p{3.4cm}|c|  }
\hline
\rowcolor{LightCyan} \hspace{37pt}{\bf Algorithm} &  {\bf Contention-Window Slots}  \\
\hline
\B & $\Theta(n\log n)$ \\
\hline
LB & $\Theta\left( \frac{n\log n}{\log\log n} \right)$  \\
\hline
LLB & $\Theta\left( \frac{n\log\log n}{\log\log\log n} \right)$  \\
\hline
STB, TSTB & $\Theta\left(n\right)$  \\
\hline
\end{tabular}
}\vspace{0pt}\caption{Known w.h.p. theoretical guarantees on CW slots for a batch of $n$ packets under BEB, LB, LLB~\cite{BenderFaHe05} and STB~\cite{Gereb-GrausT92,GreenbergL85}, and TSTB~\cite{BenderKPY18,bender:contention}.}\label{table:makespan}
\end{center}
\end{table}

For theoretical results, we measure in slots, but here a slot is an abstract unit of time. Results in the theory literature focus on the number of slots required to complete $n$ packets. However, the slots in question belong to CWs only (again, recall Section~\ref{sec:802.11}).  For example, under BEB, $n$ packets complete with high probability in $\Theta(n\log n)$ slots belonging to CWs (see~\cite{BenderFaHe05}). When evaluating backoff algorithms with a simulator, we report experimental results for this quantity using the terminology \defn{contention-window slots (CW slots)}.

\smallskip


\begin{figure*}[t]
\captionsetup[subfigure]{labelformat=empty}
\centering
\begin{subfigure}{1.0\textwidth} 
\includegraphics[width=1.0\textwidth]{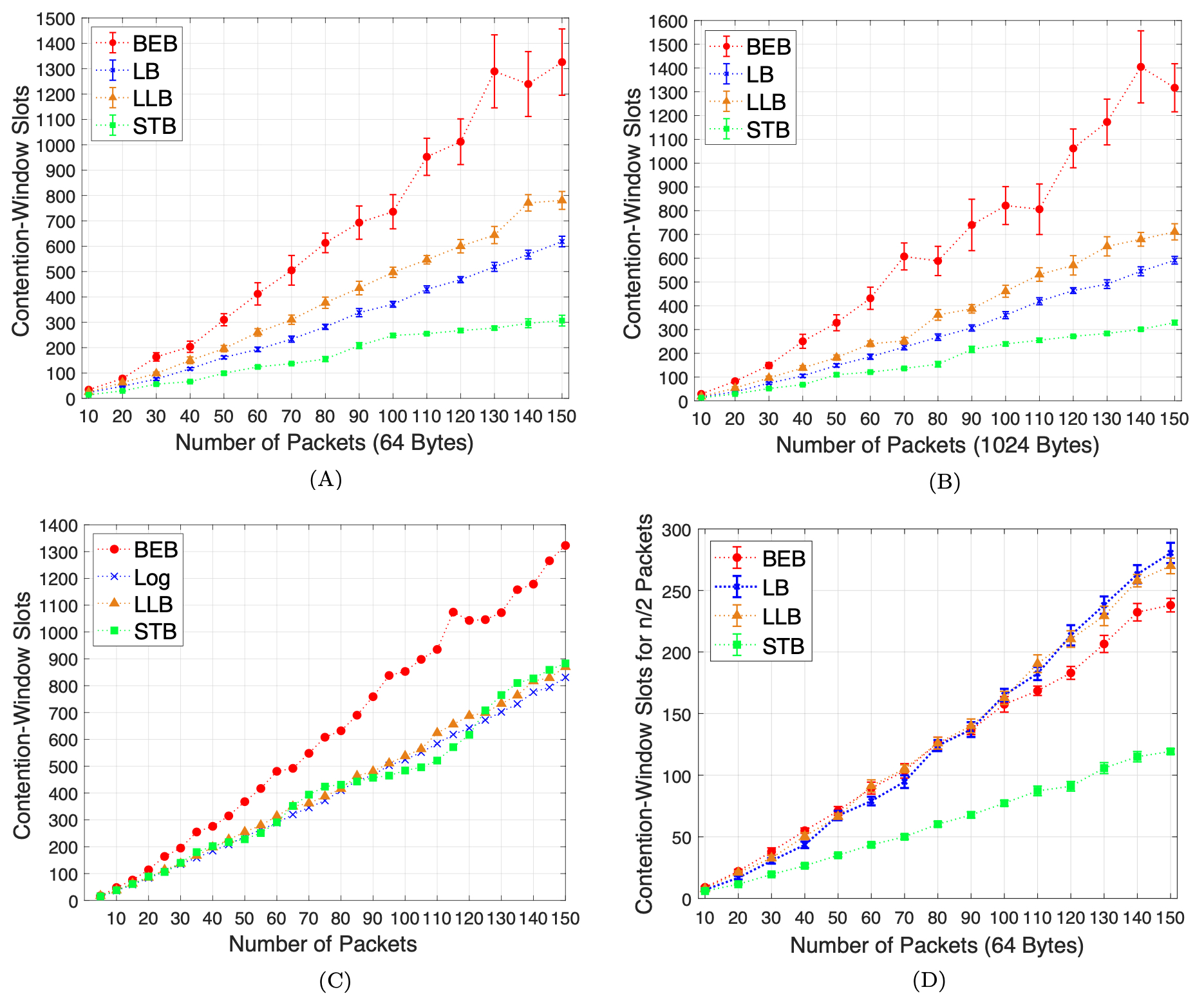} 
\end{subfigure}
\caption{Median values are reported: (A) and (B) CW slots from the \detailedSimulator with $30$ trials for each value of $n$ with $64$B and $1024$B payloads, (C) CW slots via the \theoreticalSimulator with $50$ trials for each value of $n$, (D) number of CW slots required to finish $n/2$ packets from the \detailedSimulator with $20$ trials for each value of $n$. Bars represent $95\%$ confidence intervals (omitted for (C) since they are negligible).} \label{fig:makespan}
\vspace{-12pt}
\end{figure*}


Table~\ref{table:makespan} summarizes the known with-high-probability (w.h.p.)\footnote{With probability $1-O(\frac{1}{n^{c}})$ for a tunable constant $c>1$.} guarantees on CW slots. Note that the $\Theta$-notation implies an upper and lower bound on the number of CW slots until all $n$ packets succeed.

For CW slots, LB, LLB, STB, and TSTB have superior guarantees over BEB. In particular, both STB and TSTB achieve $\Theta(n)$  CW slots, which is asymptotically optimal. Despite its similarity with STB, TSTB is worth evaluating given that its smaller number of windows may offer superior performance in practice, although we defer its examination until Section~\ref{sec:tstb}.

\subsection{Theory and Experiment}\label{sec:theory-and-experiment} \vspace{-7pt}

We begin by comparing the number of CW slots. The algorithms we investigate are designed to reduce this quantity since all slots in the theoretical model occur within some contention window. Under this metric,  LLB, LB, and STB are expected to outperform BEB. 

Throughout, when we report on performance,  we are referring to a median value of trials run with $n=150$. Percentage increases or decreases are calculated by the standard formula: $100\times(A-B)/B$, where $B$ is always the value for BEB (the ``old'' algorithm) and $A$ corresponds to a value for one of LLB, LB, or STB (the ``new'' algorithms).\vspace{-5pt}


\begin{figure*}[t!]
\captionsetup[subfigure]{labelformat=empty}
\centering
\begin{subfigure}{1.0\textwidth} 
\includegraphics[width=1.0\textwidth]{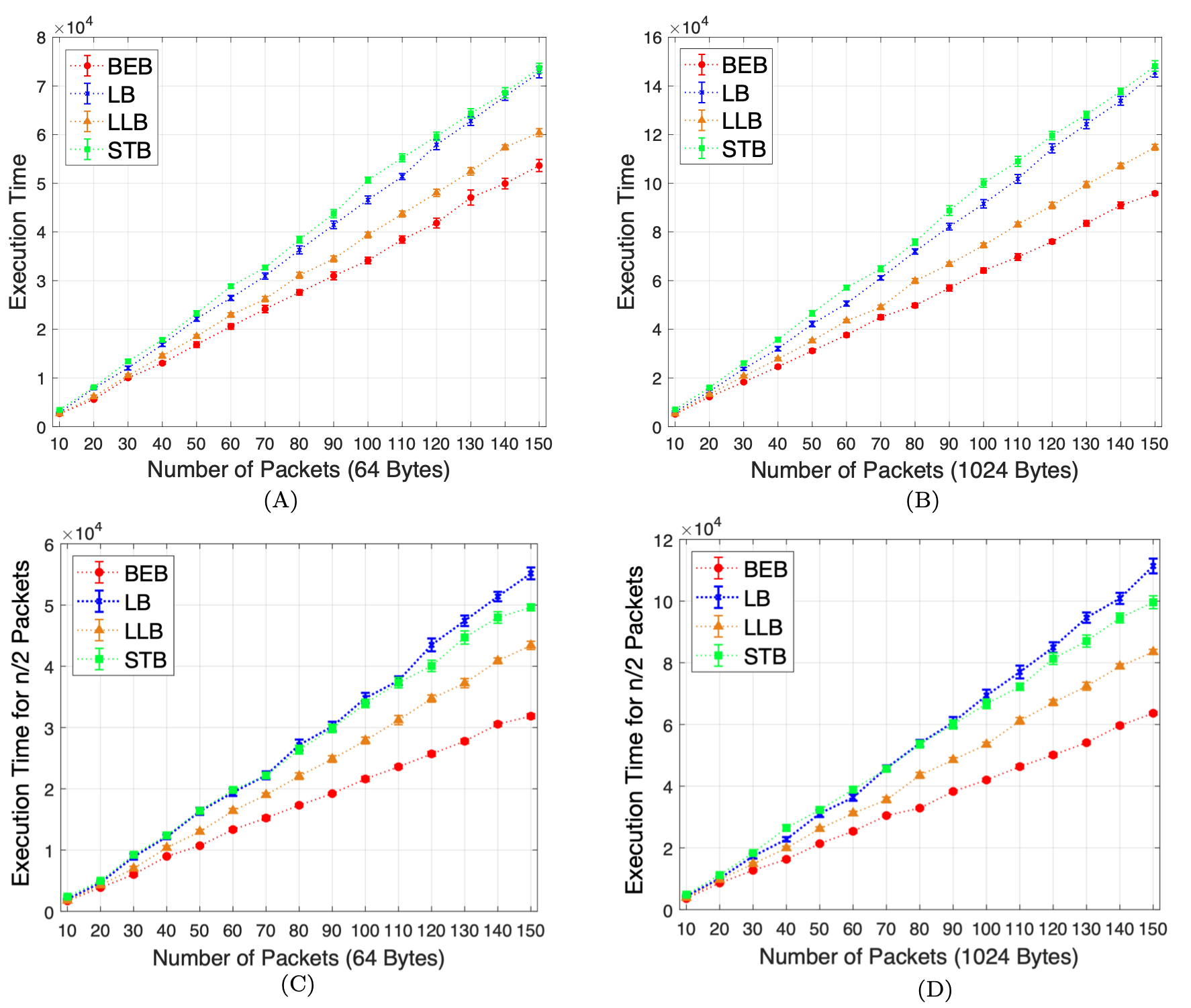} 
\end{subfigure}
\caption{The \detailedSimulator results with the median reported from $30$ trials for each value of $n$ and time measured in $\mu s$: (A) and (B) give the \totaltime for $64$B and $1024$B payloads,  (C) and (D) give the  \totaltime required to complete $n/2$ packets with a $64$B payload and $1024$B payloads.}\label{fig:total-time}  
\end{figure*}


\begin{figure*}[t]\vspace{-1cm}
\captionsetup[subfigure]{labelformat=empty}
\centering
\begin{subfigure}{1.0\textwidth} 
\includegraphics[width=1.0\textwidth]{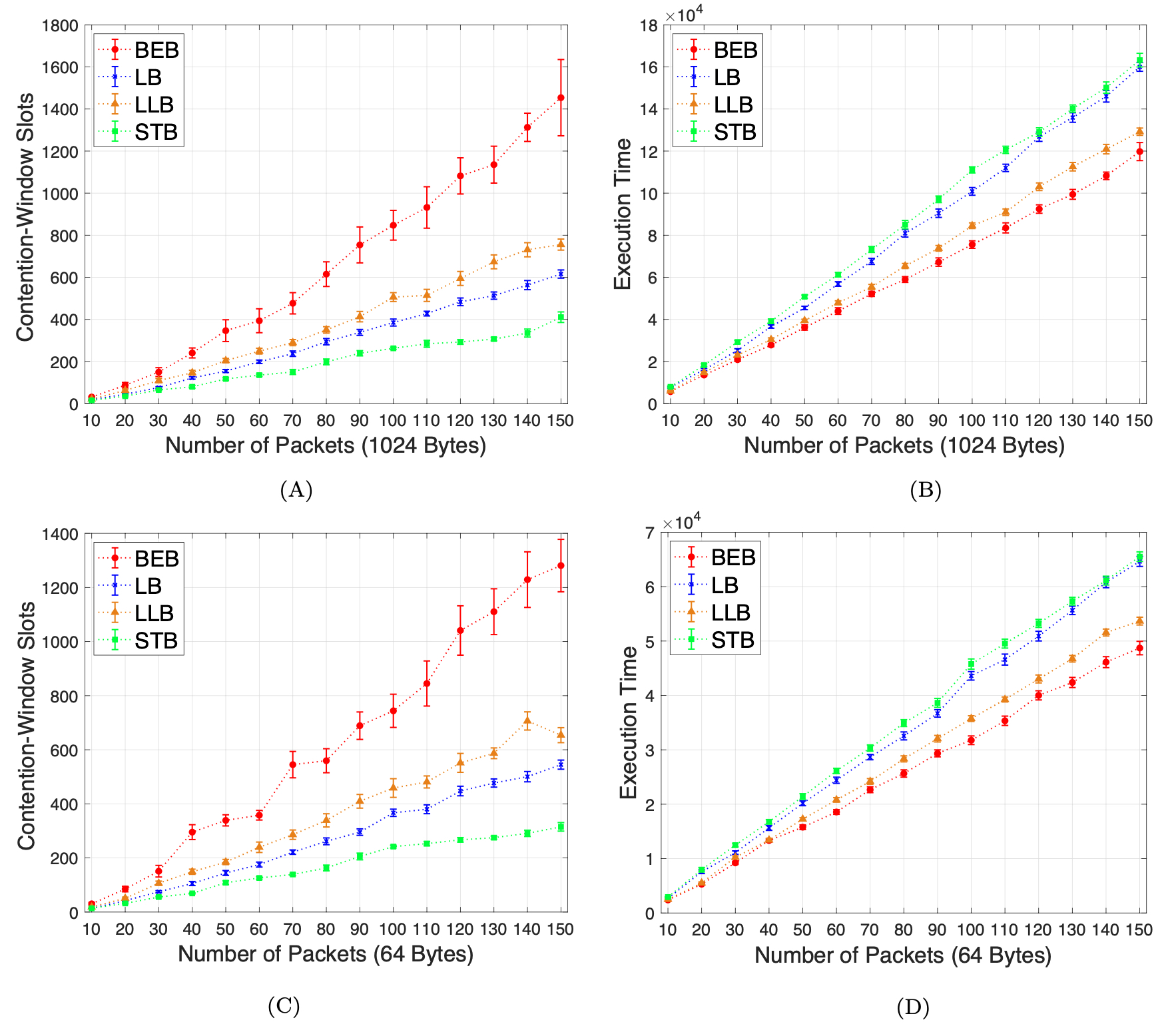} 
\end{subfigure}
\caption{Plots from the \detailedSimulator. Plots (A) and (B) give the CW slots and \totaltime, respectively,  when the slot duration is $20 \mu s$ and the SIFS duration is $10 \mu s$. Plots (C) and (D) depict CW slots and \totaltime under a random positioning of the stations, along with the AP located in the  north-east corner of the plane.} \label{fig:random-placement}\vspace{-5pt}
\end{figure*}


\subsubsection{Contention-Window Slots}\label{sec:CW-slots-exp}\vspace{-5pt}

We provide results from the \detailedSimulator using relatively small packets, with a $64$-byte (B)  payload, and larger packets, with a $1024$B payload.  

Figures~\ref{fig:makespan}(A) and \ref{fig:makespan}(B) illustrate our experimental findings with respect to CW slots. The behavior generally agrees with theoretical predictions that each of LLB, LB, and STB should outperform BEB. 

Interestingly, LLB incurs more CW slots than LB despite the former's better asymptotic guarantees. We suspect this is an artifact of hidden constants/scaling, and evidence of this is presented later in Section~\ref{sec:as-total}. 

Nevertheless,  in agreement with theory, LLB, LB, and STB demonstrate improvements over BEB, giving a respective decrease of $40.2$\%, $52.6$\%, and $76.5$\%, respectively, with a $64$B payload. This  behavior repeats with a $1024$B payload, where LLB, LB, STB demonstrate a respective decrease of $45.7$\%, $54.8$\%,  $75.1$\%.

For comparison, Figure~\ref{fig:makespan}(C) depicts CW slots derived from our \theoreticalSimulator. Reassuringly,  this roughly agrees with our results from the \detailedSimulator in terms of magnitude of values and the separation of BEB from the other algorithms; albeit, the performances of LLB, LB, and STB do not separate as cleanly in this data.

Despite BEB's inferior performance, is it possible that many packets finish faster relative to the newer algorithms, and only a few packets account for the remaining makespan? Such behavior might suggest that BEB is performing well for the majority of packets. We examine this in Figure~\ref{fig:makespan}(D) using a $64$B payload, where we  present the number of CW slots required for half of the packets to succeed.  A few observations can be made. First, BEB does not exhibit such behavior. Second, for all algorithms, the remaining $n/2$ packets are responsible for the bulk of the CW slots . Third, the plot illustrates a qualitative difference between the monotonic backoff algorithms (BEB, LLB, LB) and the {\it backon} behavior of STB.\smallskip

\begin{tcolorbox}[standard jigsaw, opacityback=0]
\noindent{\bf Result 1.~}{\it Experiments confirm theoretical predictions that LLB, LB, and STB outperform BEB with respect to CW slots.}
\end{tcolorbox}

\begin{table*}[t] \vspace{10pt}

\begin{center}

{

\begin{tabular}{  |M{1.8cm}|M{4cm}|M{4cm}|M{4cm}|  }

\hline

\rowcolor{LightCyan} {\bf Algorithm} & {\bf Sections~\ref{sec:CW-slots-exp} and~\ref{sec:totaltime}} &  \hspace{2pt} {\bf Slot duration \boldmath{$= 20 \mu s$} \newline\mbox{\hspace{6pt}}SIFS$=10\mu s$} &   {\bf Random Placement} \\

\hline

 \rowcolor{light-gray} & CW Slots~~~~~\Totaltime & CW Slots~~~~~\Totaltime & CW Slots~~~~~\Totaltime\\

\hline

{\bf LLB} & \hspace{-6pt}+40.2\%~~~~~~~~~~--12.9\% &   \hspace{-6pt}+44.4\%~~~~~~~~~~--12.8\% & \hspace{-6pt}+48.9\%~~~~~~~~~~--10.2\% \\  

{\bf LB} &  \hspace{-6pt}+52.6\%~~~~~~~~~~--36.1\% &  \hspace{-6pt}+54.8\%~~~~~~~~~~--39.2\% & \hspace{-6pt}+57.5\%~~~~~~~~~~--32.7\%\\

{\bf STB} & \hspace{-6pt}+76.5\%~~~~~~~~~~--36.9\% & \hspace{-6pt}+75.4\%~~~~~~~~~~--40.5\%   & \hspace{-6pt}+75.4\%~~~~~~~~~~--34.5\%\\

\hline

\end{tabular}

}\caption{Performance of BEB versus LB, LLB, and STB for some different parameter settings. A `+' indicates the algorithm in the first column is outperforming BEB by a percentage given by the formula in Section~\ref{sec:theory-and-experiment}, while a `--' indicates the corresponding algorithm is underperforming BEB.}\label{table:comparison}\vspace{-15pt}

\end{center}

\end{table*}


\subsubsection{\Totaltime}\label{sec:totaltime}\vspace{-5pt}

It is tempting to consider the single-batch scenario settled. However, if we focus on the \totaltime for both the $64$B and $1024$B payload sizes, then a  different picture emerges. 

The degree to which the newer algorithms outperform~\B~is erased, as depicted in Figures~\ref{fig:total-time}(A) and \ref{fig:total-time}(B). In fact, the order of performance is reversed with \totaltime ordered from least to greatest as BEB,  LLB, LB,  STB. Specifically, for $64$B payloads, LLB, LB, and STB suffer an increase of $12.9\%$, $36.1\%$, and $36.9\%$, respectively, over BEB. For $1024$B payloads, the increase is $19.6\%$, $51.6\%$, and $54.7\%$, respectively. 

Notably, the larger packet size seems to favor BEB. We revisit this observation in Section~\ref{sec:theory} when discussing a new theoretical model, and the  impact of a larger packet size.

What about the time until $n/2$ packets are successfully transmitted? Perhaps newer algorithms do better for the bulk of packets, but suffer from a few stragglers? Interestingly, Figures~\ref{fig:total-time}(C) and \ref{fig:total-time}(D) suggest that this is not the case. Indeed, for a $64$B payload, BEB performs even better over LLB, LB, and STB with the latter exhibiting an increase of $36.3\%$, $73.2\%$, $55.8\%$, respectively. Similarly, for $1024$B, the percent increase is $31.2\%$, $75.1\%$,  $56.6\%$, respectively.\medskip



\begin{tcolorbox}[standard jigsaw, opacityback=0]
\noindent{\bf Result 2.~}{\it In comparison to BEB, the \totaltime for  LLB, LB, and STB is significantly worse.}
\end{tcolorbox}
\vspace{-0pt}

These findings are troubling since, arguably, \totaltime is a more-important performance metric in practice than just CW slots. Critically, we note that this behavior is detected only through the use of the \detailedSimulator; it is not apparent from the \theoreticalSimulator.  \medskip

\noindent{\bf Robustness of Results.} How sensitive are these results to our setup? Already, we have illustrated that the discrepancy between CW slots and \totaltime holds whether we use small packets (a $64$-byte payload) or larger packets (a $1024$-byte payload). 

What if the slot time or SIFS (and, consequently, DIFS) are changed? We also executed the algorithms with a larger slot duration of $20 \mu s$  and a smaller SIFS of $10 \mu s$. These are the parameters used by 802.11g when there are devices present in the network running older WiFi standards. However, these new settings also have the benefit of reversing our default setup, which has a {\it smaller} slot duration ($9\mu s$) and a {\it larger} SIFS ($16 \mu s$). 

The resulting data is plotted in Figures~\ref{fig:random-placement}(A) and (B). The performance comparison for $n=150$ with a payload size of $1024$ bytes is provided in Table~\ref{table:comparison}, which provides a quick comparison of these results to those presented earlier in Sections~\ref{sec:CW-slots-exp} and~\ref{sec:totaltime}. The results suggest that the impact from these changes is not significant.

Perhaps the network topology has an impact? To investigate this, we ran our experiments with a different layout.  A $100m$  $\times$ $100m$ meter grid is used. The $x$-coordinate and $y$-coordinate of each station were chosen uniformly at random from $\{1, ..., 100\}$. We performed experiments with the AP in the center of the grid, and also with the AP at the north-east position of the plane (i.e., coordinate $(100,100)$) for an even more asymmetric layout. The results from this second setup are plotted in Figures~\ref{fig:random-placement}(C) and (D), and the data is reported in Table~\ref{table:comparison} for $n=150$ and payload size $64$ bytes, demonstrating that our qualitative results persist despite altering the layout.

\section{The Cost of Collisions}\label{sec:hidden}\vspace{-5pt}

We now investigate the cause of Result 2. The number of ACK timeouts per station provides an important hint. As Figure~\ref{fig:ack-timeouts} shows, relative to BEB, the newer algorithms are incurring substantially more ACK timeouts. This evidence points to collisions as the main culprit, since each ACK timeout may be considered to arise from a collision (we discuss this in Section~\ref{sec:collisions-discussion}). 

In particular, the way in which collision detection is performed means that each collision is costly in terms of time.  Recalling the overview in Section~\ref{sec:802.11},  in addition to executing over CW slots, time is incurred by transmitting the full packet and then waiting for an acknowledgement  until a timeout period. In support of our claim, we perform some back-of-the-envelope (BOTE) accounting of the delay caused by collisions. We use BEB with $n=150$ and a payload size of $64$ bytes as an example throughout our discussion below.

\vspace{-10pt}
\subsection{Back-of-the-Envelope (BOTE) Calculations}\label{sec:bote}

When a collision occurs, the corresponding station waits for an ACK timeout and then retransmits its packet; here, we try to estimate the resulting delay.  Recall  from Section~\ref{sec:experimental} that the duration of an ACK timeout is $75\mu$s.  To estimate the time required to transmit the data, we account for the \defn{transmission delay}, which is how long it takes to place all bits of the packet onto the channel. We assume that a packet of size $128$B ($64$B payload plus $64$B overhead), and conservatively, the transmission delay is roughly $\frac{1024 \mbox{\tiny~bits}}{54 \mbox{\tiny~Mbits/s} } \approx 19 \mu$s.\footnote{In practice,  $54$ Mbits/s is not achieved and so the true transmission delay is larger. Therefore, we are being conservative; the time for a (re)transmission is likely higher.} To this, we can add  the associated $20\mu$s preamble. We ignore propagation delay, since this is negligible for our setting given the relatively small distances involved.   

We break our BOTE calculation into two pieces, addressing collisions: (a) in windows less than size $n$, and (b) in windows of size $n$ and larger. For (a), we note that many slots in these windows should have a collision, and theory suggests that there are few successes until we are no smaller than, say, a constant factor of $n$. Interestingly,  Figure~\ref{fig:total-time}(C) indicates that some packets {\it are} finishing in portion (a), and we note that collision avoidance mechanisms, via carrier sensing and network allocation vectors (NAV)~\cite{kurose:computer}, may explain this.  If we conservatively assume $n/2 = 75$ slots with collisions for part (a), this translates into a duration of  $75 \cdot (19\mu s + 20\mu s + 75\mu s) = 8,550 \mu s$.  

For (b), we reason in the following way.  Over the entire execution, the maximum number of ACK timeouts for BEB---and, thus, the number of collisions---experienced by an unlucky station is $9$ according to Figure~5, and the median value is $7$. Many collisions will involve several stations, and these occur when the windows are small, which we address in portion (a). A smaller number of these collisions should occur in later windows included in portion (b) given that windows are larger, which reduces the chance of a collision.  We focus on (i) the last $n/2$ packets and (ii) collisions involving two packets. Given the median of $7$ collisions experienced per packet, consider the implications if, say, only $2$ of these collisions meet our restriction of (ii). Then, there are $75\cdot{}2$ slots where where two stations claim the same CW slot and send their packet (in full) while other stations are paused.  These collisions result in an aggregate  duration of roughly $150 \cdot (19\mu s + 20\mu s + 75\mu s) = 17,100\mu s$. 


\begin{figure}[t!]
\captionsetup[subfigure]{labelformat=empty}
\centering
\begin{subfigure}{0.65\textwidth} 
\hspace{-0pt}\includegraphics[width=0.75\textwidth, trim = 3in 6.5in 3in 7.3in, clip, scale=0.5]{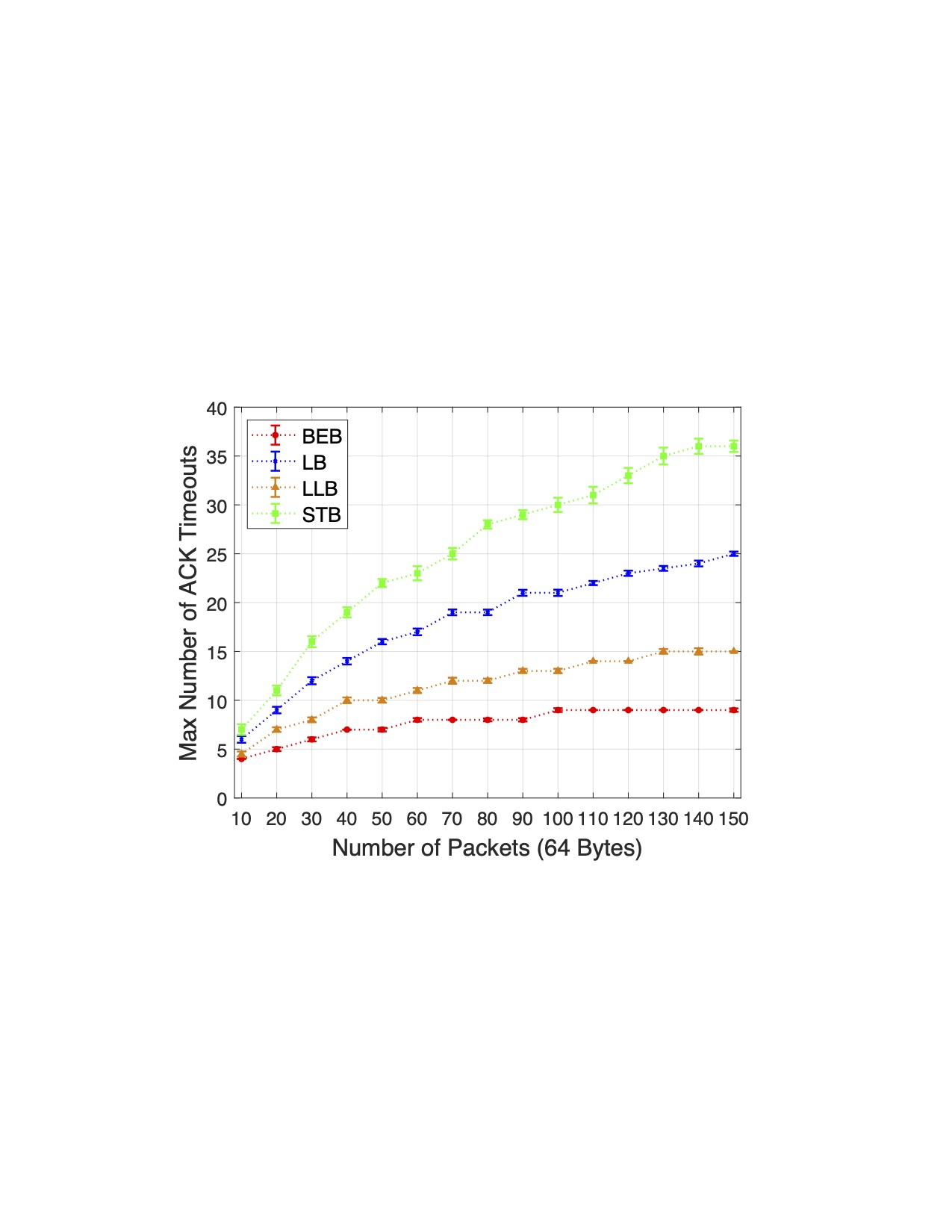}
\end{subfigure}
\vspace{-10pt}\caption{The maximum number of ACK timeouts per station over all stations with a $64$B payload from our experiments with the \detailedSimulator.  The median from $30$ trials is reported for each value of $n$.}\label{fig:ack-timeouts} 
\end{figure}


Adding the estimates from (a) and (b) together, the time due to collisions is $25,650 \mu s$. To compare this with the time consumed by CW slots, note that in Figure~\ref{fig:makespan}(A), BEB incurs $1326$ CW slots for $n=150$. Each CW slot has duration $9\mu$s, which yields a total of $9(1326) = 11,934\mu s$.  This suggests that collisions have a larger cost than CW slots.

Finally, we can incorporate the time required for the (final) successful transmission for each of the $150$ stations. We measure the time between when a station sends its packet until the time the AP sends back its ACK, which is $34 \mu s$. Given the small size of an ACK ($14$ bytes) and the small distance between any station and the AP, we treat the transmission delay and propagation delay associated with the ACK as negligible. Therefore the time for successful transmissions is of $150\cdot (19\mu s + 20\mu s  + 34 \mu s ) = 10,950 \mu s$. 

Note that summing these values yields $25,650 + 11,934 + 10,950 = 48,534  \mu s$. This seems reasonable if we compare against the measured median \totaltime value for BEB of $53,800\mu s$ in Figure~\ref{fig:total-time}(A).  In other words, these BOTE calculations provide the right order of magnitude.\smallskip
 
\subsection{Importance of Transmission Delay} The above BOTE analysis has implications for packet size or, more specifically, for time required to place all bits of the packet onto the communication channel; that is, the transmission delay.  For the $1024$B payload plus the $64$B of overhead, the transmission delay grows to $161\mu s$. 

A similar BOTE analysis for BEB yields roughly $89,000$ $\mu s$ for transmissions. By comparison, using Figure~\ref{fig:makespan}(B), the median number of CW slots remains essentially unchanged at roughly $1,300 \mu s$, contributing
 $9(1,300)\mu s$ $=$ $11,700$ $\mu s$. Therefore, transmission delay greatly dominates in this case, and the degree to which it dominates grows with the packet size.
 
 
 \medskip


\subsection{A Sanity Check} Do the simulation results bear out these BOTE calculations, or is this a ``just-so'' story? We answer this question by examining the simulation data from another different angle. We use data from with a slot duration of $20\mu s$ and the SIFS duration is $10\mu s$, and for $n=150$, we consider the total interval of time for a single trial; that is, until all $n$ packets succeed. When one or more stations is sending its packet, we count this as a {\defn{at-least-one (ALO) instance}}. We  look at the number of ALO instances over the trial and then take the median over all trials.  

This data is plotted in Figure~\ref{fig:med-total-t}. We expect  the number of ALO instances to be lowest for BEB, followed by LLB, LB, and STB in increasing order, aligning with Figure~\ref{fig:ack-timeouts}, and indeed this is what we see.  BEB has a median of $300$ ALO instances. Informally, we may consider this to be composed of $150$ final successful transmissions and $150$ retransmissions that resulted from collisions.  This implies a \totaltime of $150(19\mu s + 20\mu s + 34\mu s) + 150(19\mu s + 20\mu s + 75\mu s) =28,050 \mu s$.  

By comparison, in our BOTE calculations for BEB above, we estimated $75$ collisions (in part (a)) plus  $150$ collisions (in part (b)) for a total of $225$ collisions involving two or more stations, plus the $150$ successful transmissions. This yielded $36,600\mu s$, which is larger, but is of the same order of magnitude. 

Finally, using BEB's median of $300$ ALO instances, we can add the time of $11,934\mu s$ due to  $1326$ CW slots to arrive at a total cost of $28,050 +11,934 = 39,984\mu s$ which is lower than the BOTE calculation above of  $48, 534$, but this calculations again aligns reasonably well, and with the data for BEB with $n=150$ in Figure~\ref{fig:total-time}(A).

\begin{figure}[t]
\captionsetup[subfigure]{labelformat=empty}
\centering
\begin{subfigure}{1.0\textwidth} 
\hspace{5pt}\includegraphics[width=0.45\textwidth, trim = 0.5in 0in 1.2in 0.4in, clip, scale=0.6]{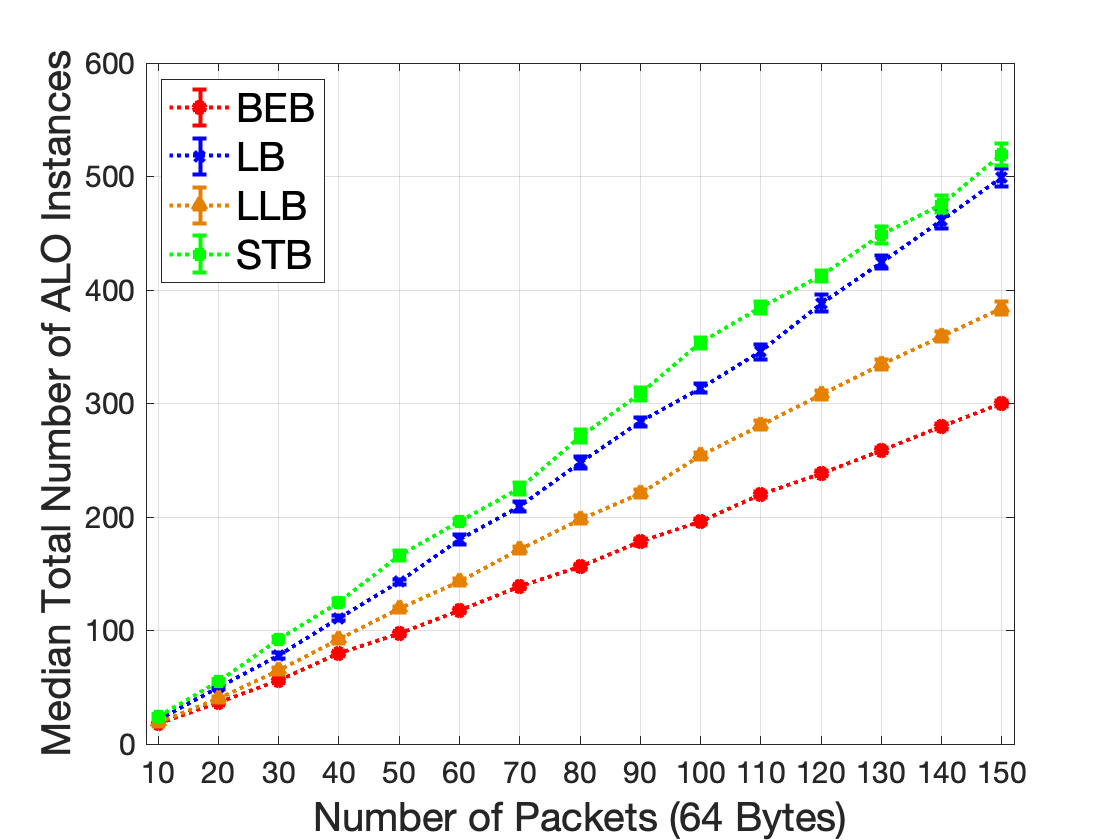} 
\end{subfigure}
\caption{ALO instances for BEB, LB, LLB, and STB. The median of $30$ trials is reported for each value of $n$ using the \detailedSimulator. }\label{fig:med-total-t} 
\end{figure}

These experimental results provide us with evidence that assumption A2 is not accurate with regards to the cost of collisions.\smallskip

\begin{tcolorbox}[standard jigsaw, opacityback=0]
\noindent{\bf Result 3.~} {\it The amount of time incurred by collisions is not properly accounted for by A2, and this cost exceeds the time incurred by CW slots.}
\end{tcolorbox}
\vspace{-15pt}

\subsection{Additional Discussion of Our Experiments}\label{sec:collisions-discussion}\vspace{-5pt}

In this section, we discuss other aspects related to our experiments. We elaborate on the correspondence between ACK timeouts and collisions in our experiments, since this plays a role in the BOTE calculations of Section~\ref{sec:bote}; in particular, we infer a collision from an ACK timeout. Additionally, we discuss why the \totaltime does not scale linearly in the number of ACK timeouts, as one might expect at first glance. Finally, we touch on whether our findings apply to situations where RTS/CTS is enabled; again, this is not often the case in practice, and this is not the focus of this work, but we present our thoughts on this issue.\medskip

\begin{figure*}[t!] 
\captionsetup[subfigure]{labelformat=empty}
\centering
\begin{subfigure}{1.0\textwidth} 
\includegraphics[width=1.0\textwidth]{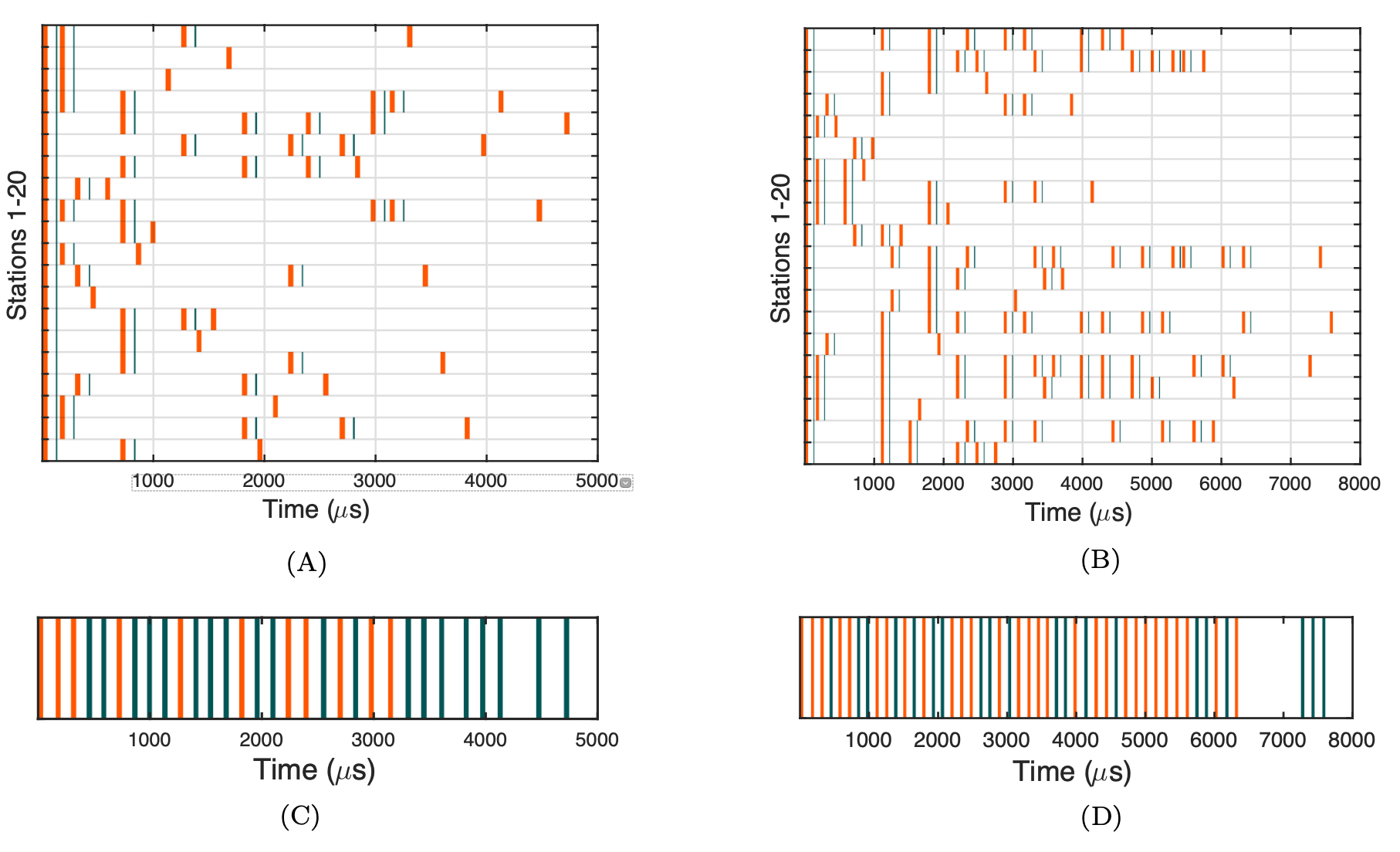} 
\end{subfigure}
\caption{Execution of BEB (A) and STB (B) with 20 stations for $64$B payload and distance $2m$. The thick orange lines correspond to the time spent transmitting the packet, and the thin blue lines indicate an ACK timeout; note that there is no such timeout in the case where the transmission is successful. For BEB and STB, respectively, plots (C) and (D) depict in red the time spent on transmitting the packet and waiting for the ACK timeout, and the green illustrates the time where a successful transmission is performed. For BEB, roughly $26\%$ of the execution time is incurred by collisions and collision detection, while for STB this number is roughly  $43\%$.
} \label{fig:beb-stb-execution}
\vspace{-0pt}
\end{figure*}

\noindent{\bf ACK Timeout \boldmath{$\approx$} Collision.} It is true that not all ACK timeouts necessarily imply that the corresponding packet suffered a collision. For example, an ACK might be lost due to signal attenuation, even if the packet was transmitted without any collision. Note that, in such a case, the sending station still diagnoses a failure, and so the same costs apply. 

However, for our simple setup, virtually all ACK failures result from a collision. This is evident from Figures~\ref{fig:beb-stb-execution}(A) and~\ref{fig:beb-stb-execution}(B) which illustrate a trial with $n=20$ under BEB and STB, using a $64$B payload. Collisions occur only when two or more stations transmit (illustratd by a thick orange line) at the same time and the result is an ACK timeout event (indicated by a thin blue line); in all other cases, the transmission is successful and the corresponding ACK is received.

Another interesting feature of these plots is that they again illustrate the difference in number of collisions experienced by BEB versus STB. Figures~\ref{fig:beb-stb-execution}(C) and~\ref{fig:beb-stb-execution}(D) condense the executions in Figures~\ref{fig:beb-stb-execution}(A) and \ref{fig:beb-stb-execution}(B), respectively,  by depicting the time over which there are no transmissions, collisions occur, or there is a successful transmission. The difference in the number of collisions is apparent. If we conservatively assume that the cost of a collision is the time between the start of the packet transmission and the time of an ACK timeout, then for BEB, this is roughly $26\%$ of the execution time, while for STB this number is $43\%$.


\medskip


\noindent{\bf \Totaltime and Number of Collisions.} We observe that the \totaltime does not grow linearly with the maximum number of ACK timeouts experienced by a station. Under LLB, an unlucky station suffers roughly $1.5\times$ the number of ACK timeouts compared to BEB, but the \totaltime of LLB is not $1.5\times$ that of BEB. 

Why? Consider $n$ stations where each collision involves only two stations. Then, there are $n/2$ slots with a collision, and each is added to the \totaltime.  In contrast, consider the opposite extreme where all $n$ stations transmit in the same slot and collide. Then, there is a single collision which adds only a single transmission delay to the \totaltime.

The number of stations involved in a single collision is larger for algorithms whose CWs grow more slowly than under BEB, such as LB and LLB.  Here, the windows are smaller for a longer period of time, so a collision typically involves many packets. For STB, a similar phenomenon is at work; the backon component yields collisions involving many stations. That is, LB, LLB, and STB are closer to the second case above. In contrast,  BEB is closer to the first as it grows its windows the fastest and does not have a backon component. Therefore, LB can have roughly $1.5\times$ as many ACK timeouts as BEB, but the \totaltime for LB is not $1.5\times$ that of BEB.\medskip

\noindent{\bf RTS/CTS.} Although it is not examined in detail in our work, we remark on the use of  RTS/CTS. When enabled, stations can experience collisions among the RTS frames (rather than among packets, for the most part). These are smaller in size, but the remainder of the \totaltime calculation remains the same, and additional time is incurred due to additional inter-frame spaces and the transmission of CTS frames. 


Ultimately, we witness the same qualitative behavior when RTS/CTS is enabled. For example, without RTS/CTS, recall from Section~\ref{sec:totaltime} that the \totaltime for LLB (BEB's closest competitor) increases by $12.9$\% and $19.6$\% for the $64$B and $1024$B, respectively, over BEB. With RTS/CTS, the respective increases are $14.5$\% and $9.5$\%.

Interestingly, the increase in \totaltime is somewhat less for the $1024$B packets. Why might this be? We believe the following qualitative behavior may explain this.  Let us make the simplifying assumption that the only collisions which occur involve an RTS frame; this is likely mostly true in our experiments. Note that RTS frames are fairly small at $20$B. 

Consider the following two cases. In Case 1, we have small data packets. Consequently, the impact of RTS collisions is significant to the \totaltime because resending an RTS takes time on the same order as the transmission delay for the data frame. Since LLB has more RTS collisions, and RTS collisions are significant, LLB will do much worse than BEB in this case.

Conversely, in Case 2,  we have large data packets. Here,  the impact of RTS is negligible compared to the time to send the data frame (which is never involved in a collision). While LLB has more RTS collisions, each one incurs little time compared to sending the data frame (which never collides). Thus, in contrast to Case 1, the degree to which LLB underperforms BEB will be less. This aligns with the notion that RTS/CTS can be helpful when packets are large; however, this does not make the other algorithms competitive with BEB.


\subsubsection{Backing Off Slowly is Bad}

The reason for the discrepancy between theory and experiment is seemingly apparent. LLB increases each successive contention window by a smaller amount than BEB; in other words, LLB is {\it backing off more slowly}, and the same is true of LB. Informally, this slower-backoff behavior is the reason behind the superior number of CW slots for LB and LLB since they linger in CWs where the contention is ``just right'' for a significant fraction of the packets to succeed. However, backing off slowly also inflicts a greater number of collisions.

Note that BEB backs off faster, jumping away from such favorable contention windows and thus incurring many empty slots. This is undesirable from the perspective of optimizing the number of CW slots. However, the result is fewer collisions. Given the empirical results, this appears to be a favorable tradeoff. We explicitly note that LLB backs off faster than LB. In this way, LLB is closer to BEB and, therefore, is not outperformed as badly in terms of \totaltime as illustrated in Figures~4(A) and 4(B). 

We highlight this observation below: \medskip

\begin{tcolorbox}[standard jigsaw, opacityback=0]
\noindent{\bf Result 4.~} {\it For algorithm design, optimizing CW slots at the expense of increased collisions is a poor design choice.}
\end{tcolorbox}


\section{A New Theoretical Model}\label{sec:theory}

We now argue for a new theoretical model to account for our experimental findings.  As before, time proceeds in abstract slots, and we denote a slot duration by {\boldmath{$s$}}.  We preserve A0 and A1 from our discussion in of Section~\ref{sec:common-model}, while modifying A2:

\begin{itemize}[leftmargin=3.5mm]
\item{\bf A0.} Each slot has length that can accommodate a packet. \medskip

\item {\bf A1.} If a single packet transmits in a slot, the packet \defn{succeeds}, but \defn{failure} occurs if two or more packets transmit simultaneously due to a \defn{collision}.  \medskip

\item {\bf A2$^*$\hspace{-3pt}.} A collision incurs a delay of $D\geq 1$ slots. 
\end{itemize}

In A2$^*$,  $D$ is a  parameter that captures the cost for a collision measured in the number of abstract slots. We argue that the \totaltime for a backoff algorithm $A$ is the sum of (i) the time spent executing backoff in contention windows, plus (ii) the time incurred by collisions: 
$$ \mathcal{W}_A\cdot{s} + \mathcal{C}_A\cdot D \cdot s  $$
\noindent where  $\mathcal{W}_A$ is the number of CW slots, and $\mathcal{C}_A$ is the number of slots with a collision. We can represent this more naturally as a number of slots, denoted by {\boldmath{$\mathcal{T}_A$}}:
$$\mathcal{T}_A =   \mathcal{W}_A  + \mathcal{C}_A\cdot D  $$

Our revised theoretical model remains simple, but we will later argue (Section~\ref{sec:as-total}) that it explains the behavior observed in our experiments with the \detailedSimulator.  

In this section, we analyze BEB, LB, LLB, and STB using our model to derive new theoretical bounds on $\mathcal{T}_A$. For this analysis, we note that prior  results in the literature (recall Table~\ref{table:makespan}) already establish $\mathcal{W}_A$. Therefore, we are interested in $\mathcal{C}_A$, for which we derive asymptotically tight bounds. To be clear, we will prove the following:\medskip\smallskip
\begin{itemize}
\item[{$\bullet$}]{\bf BEB (Section~\ref{sec:beb-analysis}).} An upper and lower bound on $\mathcal{C}_{\mbox{\tiny BEB}}$. It has been shown that $\mathcal{W}_{\mbox{\tiny BEB}}= O(n\log n)$~\cite{BenderFaHe05}, and we will show that $\mathcal{C}_{\mbox{\tiny BEB}} = \Theta(n)$. \medskip\smallskip

\item[{$\bullet$}]{\bf LLB and LB (Sections~\ref{sec:llb} and~\ref{sec:lb}).} A lower bound for both algorithms. Specifically, we will prove that $\mathcal{C}_{\mbox{\tiny LLB}} $ $=$ $\Omega( \frac{n\log\log n}{\log\log\log n})$. Since $\mathcal{W}_{\mbox{\tiny LLB}}\hspace{-3pt}=\hspace{-2pt}O( \frac{n\log\log n}{\log\log\log n})$ \cite{BenderFaHe05}, this will imply that $\mathcal{C}_{\mbox{\tiny LLB}} = \Theta( \frac{n\log\log n}{\log\log\log n})$. Similarly, for LB, we will show  $\mathcal{C}_{\mbox{\tiny LB}} = \Omega(\frac{n\log n}{\log\log n})$, which is asymptotically tight for the same reason.\medskip\smallskip

\item[{$\bullet$}]{\bf STB (Section~\ref{sec:stb-analysis})}. We will show that $\mathcal{C}_{\mbox{\tiny STB}} = \Omega(n)$. Since $\mathcal{W}_{\mbox{\tiny STB}}= O(n)$~\cite{Gereb-GrausT92,GreenbergL85}, this will imply that $\mathcal{C}_{\mbox{\tiny STB}} = \Theta(n)$.\medskip

\end{itemize}

To derive these bounds, we follow in the footsteps of Bender et al.~\cite{BenderFaHe05}, who previously provided bounds on $\mathcal{W}_A$. The new asymptotic bounds on  $\mathcal{C}_A$ match $\mathcal{W}_A$, except for the case of BEB. However, these new results still require  careful analysis, as they do not immediately follow from $\mathcal{W}_A$, as BEB exemplifies.

Finally, in Section~\ref{sec:as-total}, we will combine the old bounds on  $\mathcal{W}_A$ with our  bounds on  $\mathcal{C}_A$ to arrive at bounds on $\mathcal{T}_A$. We discuss the implications of $\mathcal{T}_A$ and argue, given experimental results, that A2$^*$ should be used instead of A2. 


\subsection{Analyzing $\mathcal{C}_A$}\label{sec:analyzing-collisions}

In order to provide additional support for our empirical findings, we derive asymptotic bounds on $\mathcal{C}_A$. Our arguments are couched in terms of packets and slots, but what follows is a balls-into-bins analysis, where balls correspond to packets, and slots correspond to bins. To bound $\mathcal{C}_A$, we are interested in the number of bins---where bins correspond to the slots in a CW---that contain two or more balls, since this is the equivalent of a collision. Throughout our analysis, we use the terminology of slots and bins interchangeably. Finally, all of our results hold given that $n$ is sufficiently large.

We make use of the following well-known inequalities.\footnote{For example, these inequalities are established in Lemma 3.3 by Richa et al.~\cite{richa:jamming4}.} 

\begin{fact}\label{fact:taylor-lower}
For any $0 \leq x<1$,   $e^{-x/(1-x)} \leq 1 - x$.
\end{fact}

\begin{fact}\label{fact:taylor-upper}
For any $x$,   $1 - x \leq e^{-x}$.
\end{fact}

The bounds that we derive in this section hold \defn{with high probability (w.h.p.)}; that is, with probability at least $1 - 1/n^c$ for some tunable constant $c>1$. Throughout our analysis, we make use of the following previously-known concentration result:

\begin{theorem}\label{thm:MOBD}
(Method of Bounded Differences, Corollary 5.2 in~\cite{dubhashi:concentration})  Let $f$ be a function of independent variables $X_1, ..., X_N$ such that for any $b,b'$ it holds that $|f(X_1, ..., X_i = b, ..., X_N ) - f(X_1, ..., X_i = b', ..., X_N)| \leq c_i$ for $i=1, ..., N$. Then, the following holds:\vspace{-5pt}
$$Pr( f > E[f] + t)  \leq  e^{-2t^2/(\sum_{i} c_i^2)}\vspace{-3pt}$$ 
$$Pr( f < E[f] - t) \leq e^{-2t^2/(\sum_{i} c_i^2)}.$$
\end{theorem}
\vspace{-3pt}
\noindent{}Note that this concentration result holds for random variables that may be dependent.  The following result follows directly, and we factor it out here since we use it in several of our arguments:

\begin{corollary}\label{cor:whp}
Let $Z$ be the number of collisions in a window of size $w$. Then:
$$Pr\left( Z > E[Z] + t\right)  \leq  e^{-2t^2/w}\vspace{-3pt}$$
$$Pr\left( Z < E[Z] - t\right)  \leq  e^{-2t^2/w}.$$
\end{corollary}
\begin{proof}
Let $Z$  denote the number of collisions in window $i$ of a windowed backoff algorithm. We may view $Z$ as a function $Z(X_1, ..., X_n)$, where for each $k\in\{1, ..., n\}$, $X_k$ is a random variable with a value in $\{1, ..., w\}$ corresponding to the bin in which the $k^{\mbox{\tiny th}}$ ball lands. Note that $Z$ is certainly a function of these $X_k$ variables, although we do not specify the function itself, nor do we need to in order to use Theorem~\ref{thm:MOBD}. 

The $X_k$ variables are independent, since the bin in which ball $k$ lands is independent of the bin in which any other ball lands. The $X_k$ variables satisfy:\vspace{-0pt}
$$| Z(X_1, ..., X_{i} = b, ...,  X_{w}) - Z(X_1, ..., X_{i} = b', ...,  X_{w})| \leq 1 $$

\noindent since if a ball is moved from one bin into another, then the number of collisions either increases or decreases by at most $1$.

By Theorem~\ref{thm:MOBD}, w.h.p. $Z$ is tightly bounded to its expectation:\vspace{-5pt}

\begin{eqnarray*}
Pr\left( Z > E[Z] + t\right)  &\leq &  e^{-2t^2/(\sum_{w} 1)}\\
&= & e^{-2t^2/w}
\end{eqnarray*}
\noindent and:
\begin{eqnarray*}
Pr\left( Z < E[Z] - t\right)  &\leq &  e^{-2t^2/(\sum_{w} 1)}\\
&= & e^{-2t^2/w}
\end{eqnarray*}
which completes the argument.\qed
\end{proof}
 
A similar result holds for the number of successes in a window; equivalently, the number of bins that contain a single ball.

\begin{corollary}\label{cor:successes}
Let $Z$ be the number of successes in a window of size $w$. Then:
$$Pr\left( Z > E[Z] + t\right)  \leq  e^{-2t^2/w}\vspace{-3pt}$$
$$Pr\left( Z < E[Z] - t\right)  \leq  e^{-2t^2/w}.$$
\end{corollary}

The proof for Corollary~\ref{cor:successes} is nearly identical to that of Corollary~\ref{cor:whp}, and so we omit it here.  


\subsection{Upper and Lower Bounding Collisions in BEB}\label{sec:beb-analysis}

\begin{theorem}
For a single batch of $n$ packets, with high probability, $\mathcal{C}_{\mbox{\tiny BEB}} = \Theta(n)$. 
\end{theorem}
\begin{proof}
First, we show that $\mathcal{C}_{\mbox{\tiny BEB}} = \Omega(n)$. Note that for the first windows of size $1, 2, ...,  n/16$, there are less than  $n/8$ slots by the sum of a geometric series, and so less than $n/8$ packets can succeed over these windows. Consider the next window of size $n/8$, and treat this as a balls-in-bins problem with $n/8$ bins and at least $(7/8)n$ balls. We will count the number of bins that contain $2$ balls, since this will then allow us to show a lower bound on the number of collisions. 

Let the indicator random variable $Y_j=1$ if bin $j$ contains $2$ balls; otherwise, $Y_j=0$. It is easy to see that $Pr(Y_j=1)$ is a constant. Letting $Y=\sum_j Y_j$, we have $E[Y] = \Omega(n)$. By Corollary~\ref{cor:whp}:
\begin{eqnarray*} 
Pr\left( Y < E[Y] - t\right)  &\leq&  e^{-2t^2/(n/8)}\\
&=& 2^{-\Theta(n)}
\end{eqnarray*}
\noindent for $t=k'n$, for some suitably small constant $k'>0$, which yields the lower bound.

For the upper bound, in the execution of~\B, consider a contention window of size $n 2^i$ for an integer $i\geq 1$; let the windows be indexed by $i$. Note that over all  windows up to and including size $n$, we have $O(n)$ collisions since there are $O(n)$ slots by the sum of a geometric series. 

Let the indicator random variable $Z_j=1$ if slot $j$ in window $i$  contains a collision; $Z_j=0$ otherwise. Considering the corresponding balls-into-bins problem,  the probability of a collision is equal to the probability that $2$ or more balls are dropped into this bin. Pessimistically, assume exactly $n$ balls are dropped in each consecutive window $i$, since this only increases the probability of a collision; in actuality, packets finish over these windows and reduce the probability of collisions. 
\begin{eqnarray*}
Pr(Z_j=1)  &= & \sum_{\ell=2}^{n} \binom{n}{\ell}\left(\frac{1}{n2^i}\right)^{\ell}  \left(1- \frac{1}{n2^i}\right)^{n-\ell}\\
& = & \sum_{\ell=2}^{n} O\left(\frac{1}{2^{i\ell}}\right)  \mbox{~~by Fact~\ref{fact:taylor-upper}}\\
& = & O\left(\frac{1}{2^{2i}}\right)
\end{eqnarray*}

 Let $Z^i=\sum_{j=1}^{n2^i } Z_j$ be an upper bound on  the number of collisions in window $i$. By linearity of expectation:\vspace{-4pt}
$$E[Z^i] = \sum_{j=1}^{n2^i } E\left[Z_j\right] =  n2^i \cdot O\left(\frac{1}{2^{2i}}\right)= O\left(\frac{n}{2^{i}}\right)$$

\noindent There are dependencies between the $Z_j$ variables. To handle this complication,  we employ Corollary~\ref{cor:whp} to argue that $Z^i$ is tightly bounded to its expectation; that is, for some constant $d>0$, we have:
$$Pr\left( Z^i > \frac{dn}{2^i} + t\right) \leq e^{-2t^2/(n2^i)}$$

\noindent By~\cite{BenderFaHe05,zhou2019singletons}, w.h.p. \B~finishes within $O(\lg\lg n)$ windows subsequent to a window of size $n$; therefore, $i \leq d'\lg\lg n$ for some constant $d'>0$. Setting $t = n/2^i$, and plugging into the above:\vspace{-5pt}

\begin{eqnarray*}
Pr\left( Z^i > \frac{dn}{2^i} + \frac{n}{2^i}\right)  &\leq &  e^{-2n/2^{3i}}\\
& \leq & e^{-2n/\lg^{3d'}n}
\end{eqnarray*}

\noindent Taking a union bound over $i \leq d'\lg\lg n$ windows, it follows that w.h.p. the number of collisions under  BEB is upper bounded by $\sum_{i=0}^{d'\lg\lg n} Z^i = \sum_{i=1}^{d'\lg\lg n} O\left(\frac{n}{2^i}\right) = O(n)$.\qed
\end{proof}


\subsection{Lower Bounding Collisions in LLB}\label{sec:llb}

In this section, we derive a lower bound on the number of collisions experienced by LLB. For ease of analysis, we assume a larger initial (but still constant) window size for LLB; we defer discussion of this for now. 

To prove our lower bound, we consider a variant of LLB that we denote by LLB*, and which executes as follows: a contention window of size $w_i$ is used consecutively $(1/2)\lg\lg(w_i)$ times, and then $w_{i+1} = 2w_i$ for $i\geq 0$.\footnote{We highlight that the constant $1/2$ is not special. As the proof of Lemma~\ref{lem:LLB-star} shows, any positive constant strictly less than $\ln(2)$ will suffice. A similar algorithm is examined by Bender et al.~\cite{BenderFaHe05} in order to obtain an upper bound on the makespan of LLB.}  Again, we defer discussion of the size of the initial window until after Lemma~\ref{lem:LLB-star}.

We argue that the $i^{th}$ contention window under LLB* is larger than the $i^{th}$ contention window under LLB, assuming that we set the initial window size for each to be a sufficiently large constant. This  implies that a lower bound on the number of collisions experienced by LLB* is a lower bound on the number of collisions experienced by LLB. The remainder of our analysis is devoted to showing a lower bound on the number of collisions experienced by LLB*.

\begin{lemma}\label{lem:LLB-star}
The window size doubles in fewer windows under LLB* than under LLB for windows of size greater than  $2^{2^{2/(\ln(2)-1/2)}}$.
\end{lemma}
\begin{proof}
Under LLB, consider any window $W_i$ of size $w_i$ for $i\geq 0$. How many windows are required until the window size becomes at least $2w_i$? It holds that:
\begin{eqnarray*}
w_{i+k} & = & \left(\prod_{j=0}^{k-1} \left( 1 + \frac{1}{\lg\lg w_{i+j}} \right) \right)w_i\\
&\leq& \left(1 + \frac{1}{\lg\lg w_i}\right)^k w_i \\
&\leq & e^{k/\lg\lg w_i} w_i 
\end{eqnarray*}
\noindent where the first line follows by the specification of LLB, the second line follows since $w_i$ is the smallest window size, and the third line follows from Fact~\ref{fact:taylor-upper}. Solving for $e^{k/\lg\lg w_i}  \geq 2$ yields $k \geq  \ln(2) \lg\lg w_i$. That is, at least:
\begin{eqnarray}
\lfloor \ln(2) \lg\lg w_i\rfloor \geq   \ln(2) \lg\lg w_i - 1\label{eq:normal}
\end{eqnarray}

\noindent  windows are required to double the window size under LLB. In contrast, LLB* requires at most:
\begin{eqnarray}
\lceil(1/2)\lg\lg w_i\rceil \leq (1/2)\lg\lg w_i + 1\label{eq:star}
\end{eqnarray}
\noindent windows to double.  Solving:
$$(1/2)\lg\lg(w_i) + 1 < \ln(2)\lg\lg(w_i) - 1 $$
\noindent we find that for $w_i > 2^{2^{2/(\ln(2)-1/2)}} =O(1)$,  the upper bound in Equation~\ref{eq:star} is less than the lower bound in Equation~\ref{eq:normal}. That is, after this window size is reached, LLB* will be doubling its window size in fewer windows than LLB.\qed
\end{proof}

\noindent{\bf Setting Initial Window Size.} By Lemma~\ref{lem:LLB-star}, beyond a window size of $w_{i^*} > 2^{2^{2/(\ln(2)-1/2)}}$, LLB* requires fewer windows to double its window size compared to the window of the same index under LLB. We set LLB's initial window size to be a constant at least as large as $w_{i^*}$. Then, we set LLB*'s initial window size to be at least twice LLB's initial window size. This ensures that LLB*'s window will always be larger than the window of the same index under LLB.\medskip

\subsubsection{A Lower Bound on Collisions for LLB*}

For our analysis of LLB*,  we focus on a window of size  $cn/\lg\lg\lg n$, for some sufficiently small constant $c>0$. Our argument proceeds as follows. Given $\epsilon n$ packets, for any positive constant $\epsilon \leq 1$, we first demonstrate in Lemma~\ref{lem:LLB-succ} an upper bound of $O(n/(\lg\lg n)^{d})$ successes  in any execution over a window of size $cn/\lg\lg\lg n$, where $d>1$ is some constant depending on $\epsilon$ and $c$. Next, in Lemma~\ref{lem:less-n}, we show that $\Theta(n)$ packets survive until the critical window size is reached. Lemmas~\ref{lem:LLB-succ} and~\ref{lem:less-n} imply that there are $\Theta(n)$ packets remaining for $\Omega(\lg\lg n)$ executions of a window of size $cn/\lg\lg\lg n$. Finally, in Theorem~\ref{claim:LLB},  we prove that for each such execution, $\Omega(n/\lg\lg\lg n)$ collisions occur, and this yields a total of $\Omega(n\lg\lg n/\lg\lg\lg n)$ collisions. \vspace{-2pt}

\begin{lemma}\label{lem:LLB-succ}
For any constant $\epsilon\leq 1$, consider $\epsilon n$ packets executing LLB* over a window of size  $cn/\lg\lg\lg n$ for a positive constant $c<\epsilon\lg(e)/2$. With high probability, $O(n/(\lg\lg n)^{d})$ packets succeed in the window for a constant $d>1$ depending on $\epsilon$ and $c$.\vspace{-3pt}
\end{lemma}
\begin{proof}
Let $Y_j=1$ if slot $j$ contains a single packet; otherwise $Y_j=0$. We have: \vspace{-5pt}
\begin{eqnarray}
Pr(Y_j=1) & = & \binom{\epsilon n}{1}\left(\frac{\lg\lg\lg n}{cn}\right) \hspace{-3pt} \left(1- \frac{\lg\lg\lg n}{cn}\right)^{\epsilon n-1} \label{eqn:llb1}\\
& \leq &  \frac{\epsilon\lg\lg\lg n}{c} e^{-\frac{ (\epsilon n - 1)\lg\lg\lg n}{cn}} \label{e:a}\\
& \leq &  \frac{\epsilon\lg\lg\lg n}{c} e^{-\frac{  (1/2)\epsilon n \lg\lg\lg n}{cn}}  \label{e:b}\\
& \leq &  \frac{\epsilon\lg\lg\lg n}{c} e^{-\frac{ \epsilon  \lg\lg\lg n}{2c}} \\
&\leq & \frac{\epsilon\lg\lg\lg n}{c(\lg\lg n)^{\frac{\epsilon\lg(e)}{2c}}}\\
&=& O\left(\frac{1}{(\lg\lg n)^{d}}\right) \mbox{~for a constant $d>1$}\label{eqn:llb2}
\end{eqnarray}
\noindent where Equation~\ref{e:a} line follows from Fact~\ref{fact:taylor-upper}, Equation~\ref{e:b} follows from $\epsilon n\geq 2$ which is true for sufficiently large $n$, and  Equation~\ref{eqn:llb2} follows from noting that $\epsilon\lg(e)/(2c) > 1$ for any positive constant $\epsilon\leq 1$ so long as the positive constant $c< \epsilon \lg(e)/2$.

Let $Y=\sum_{j=1}^{cn/\lg\lg\lg n} Y_j$ be the number of successes over the entire window. From the above, we have:\vspace{-5pt}
$$E[Y] = O\left(\frac{n}{(\lg\lg n)^{d} \lg\lg\lg n}\right) $$

\noindent By Corollary~\ref{cor:successes}:
\begin{eqnarray*}
Pr\left( Y > E[Y] +  t \right)  & \leq & e^{-\frac{2t^2 \lg\lg\lg n}{cn} }
\end{eqnarray*}
and letting $t = \sqrt{n\ln n}$ is sufficient to prove that the number of successes in the window is $O\left(\frac{n}{(\lg\lg n)^{d} \lg\lg\lg n}\right)$ with high probability. Noting that: 
$$(\lg\lg n)^{d} \lg\lg\lg n = (\lg\lg n)^{d+ \frac{\lg\lg\lg\lg n}{\lg\lg\lg n}} =  (\lg\lg n)^{d+ o(1)}$$
completes the proof.\qed
\end{proof}

We now argue that, for LLB*, very few packets succeed prior to a window of size $cn/\lg\lg\lg n$.

\begin{lemma}\label{lem:less-n}
Consider a single batch of $n$ packets executing LLB*. With high probability, $o(n)$ packets succeed prior to the first window of size $cn/\lg\lg\lg n$ for a sufficiently small  constant $c>0$.
\end{lemma}
\begin{proof}
To begin, note that w.h.p. no packet finishes in any slot up until the end of all windows of size $n/(4\ln n)$. To see this, note that for any such slot, the probability that it contains a single packet is at most: 
\begin{eqnarray}
\binom{n}{1}\left(  \frac{4\ln n}{n}\right) \left(1 - \frac{4\ln n}{n}\right)^{n-1} & \leq & \frac{4\ln n}{e^{4(n-1)\ln n/n}}\label{eqn:less-n}\\
&=& \frac{4\ln n}{n^{4(n-1)/n}}\\
& \leq & \frac{4\ln n}{n^3}\\
& \leq & \frac{1}{n^2}\label{eqn:llb3}
\end{eqnarray}
\noindent where the first line follows from Fact~\ref{fact:taylor-upper}, and the third line follows from $n\geq 4$. Taking a union bound over all $O(n)$ slots up to this point yields a probability that any packet succeeds is at most $O(1/n).$

Now, we consider how many packets can succeed between those windows with size  $n/(4\ln n)$ and $cn/\lg\lg\lg n$. Starting with a window of size  $n/(4\ln n)$, how many intervening {\it unique-sized} windows exist before reaching the first window of size $cn/\lg\lg\lg n$? This is given by solving for $i$ in the following:
$$\frac{2^i n}{4\ln n} \leq \frac{cn}{\lg\lg\lg n}$$
yielding $i\leq \lg(4c\ln n) - \lg\lg\lg\lg(n) = k\lg\lg n$ for some positive constant $k$.

Pessimistically assume each such intervening window has size $cn/\lg\lg\lg n$ (this can only reduce the number of collisions). By the above argument, we have $\epsilon n$ packets for some constant $\epsilon>0$, and for $c$ sufficiently small, and so Lemma~\ref{lem:LLB-succ} guarantees w.h.p. that each window results in  $O(n/(\lg\lg n)^{d})$ successful packets for some constant $d>1$. 

Under LLB*, each such intervening window executes $O(\lg\lg n)$ times. Therefore, the total number of packets finished over these windows is: $$O\left(\frac{n}{(\lg\lg n)^{d-1}}\right) = o(n)$$
\noindent since $d>1$, and this completes the argument.\qed
\end{proof}

\begin{theorem}\label{claim:LLB}
For a single batch of $n$ packets, with high probability, $\mathcal{C}_{\mbox{\tiny LLB*}} = \Omega\left(\frac{n\lg\lg n}{\lg\lg\lg n}\right)$. 
\end{theorem}
\begin{proof}
We focus on a window of size $cn/\lg\lg\lg n$ for a sufficiently small constant $c>0$. Conservatively, we do not count collisions prior to this window (counting these can only improve our result).

By Lemma~\ref{lem:less-n}, w.h.p. the number of packets that succeed prior to reaching the first window of size $cn/\lg\lg\lg n$ is $o(n)$. Therefore, at the start of this window, we have $\epsilon n$ packets for some positive constant $\epsilon$. Define an indicator random variable $X_j$ such that $X_j=1$ if slot $j$ contains a collision; otherwise, $X_j=0$. Then, we have $Pr(X_j=1)$:\vspace{-5pt} 
\begin{eqnarray}
& = & 1 - \sum_{k=0}^1 \binom{\epsilon n}{k}\hspace{0pt}\left(\frac{\lg\lg\lg n}{cn}\right)^k \hspace{-3pt} \left(1- \frac{\lg\lg\lg n}{cn}\right)^{\epsilon n-k}\label{eqn:llb4}\\
& \geq & 1 - \left( 1-\frac{\lg\lg\lg n}{cn}\right)^{\epsilon n} \hspace{-8pt}- \frac{\epsilon \lg\lg\lg n}{c}\left(1-\frac{\lg\lg\lg n}{cn} \right)^{\epsilon n - 1 } \nonumber \\
& \geq & 1-  \left( 1-\frac{\lg\lg\lg n}{cn}\right)^{\epsilon n}  \left( 1 + \frac{ \epsilon\lg\lg\lg n / c}{1- \lg\lg\lg n/(cn)}  \right)\nonumber \\
& \geq & 1-  \left( 1-\frac{\lg\lg\lg n}{cn}\right)^{\epsilon n}  \left( 1 + \frac{2 \epsilon\lg\lg\lg n}{c}   \right) \nonumber\\
& \geq & 1 - O \left( \frac{\epsilon\lg\lg\lg n}{ c\left( \lg\lg n \right)^{\epsilon\lg(e)/c}} \right) \nonumber\\
& = & \Omega(1) \label{eqn:llb5}
\end{eqnarray}
\noindent where the fifth line follows from Fact~\ref{fact:taylor-upper} (with $x = \lg\lg\lg(n)/(cn))$.

Let $X = \sum_{j=1}^{cn/\lg\lg\lg n} X_j$. By linearity of expectation, the expected number of collisions over the contention window $W$ is:\vspace{-3pt}
$$E[X] = \sum_j E[X_j] = \Omega\left(\frac{n}{\lg\lg\lg n}\right)$$
\vspace{-10pt}

\noindent By Corollary~\ref{cor:whp}:
$$Pr\left( X < E[X] -  t \right)   =  e^{-\frac{2t^2 \lg\lg\lg n}{cn} }$$
\noindent and letting $t = \sqrt{n\ln n}$ is sufficient to prove that the number of collisions in the window is $\Omega\left(\frac{n}{\lg\lg\lg n}\right)$ with high probability.

By specification of LLB*, the window $w$ is executed $(1/2)\lg\lg(cn/\lg\lg\lg n) = \Omega(\lg\lg n)$ times. By Lemma~\ref{lem:LLB-succ}, $O\left(\frac{n}{(\log\log n)^d}\right)$ packets are succeeding in each such execution for some constant $d>1$, assuming the constant $c>0$ is sufficiently small. Thus, there will be at least $\Omega(n)$ packets remaining after each of the  $\Omega(\lg\lg n)$ executions of this window size. This means we can apply the above lower bound on the number of collisions per window of size $cn/\lg\lg\lg n$, and so w.h.p. there are a total of $\Omega(n/\lg\lg\lg n) \cdot \Omega(\lg\lg n)$ collisions.\qed
\end{proof} 

\noindent We make the following observation about the argument above. Since a ``small'' $O(n/(\log\log n)^d)$  number of packets are succeeding in each such execution of the window of size $cn/\lg\lg\lg n$, it may appear that proving a stronger lower bound is possible by proceeding to larger and larger windows in the execution of LLB*. However, we highlight the dependency between $d$, $c$, and $\epsilon$; that is,  for $d>1$ to hold, we require $c<\epsilon\lg(e)/2 < 0.73 \epsilon \leq 0.73$, and thus we cannot show a larger number of collisions by proceeding as suggested.

For completeness, we state the following result for LLB (rather than LLB*), which follows from Theorem~\ref{claim:LLB} and the upper bound of $\mathcal{W}_{\mbox{\tiny LLB}} = O\left(\frac{n\lg\lg n}{\lg\lg\lg n}\right)$ slots in~\cite{BenderFaHe05}.

\begin{corollary}\label{claim:LLB-final}
For a single batch of $n$ packets, with high probability, $\mathcal{C}_{\mbox{\tiny LLB}} = \Theta\left(\frac{n\lg\lg n}{\lg\lg\lg n}\right)$. 
\end{corollary}

\vspace{-15pt}
\subsection{A Lower Bound on Collisions for LB}\label{sec:lb}

Finally, we turn to the analysis of LB. Here, we define the algorithm LB* where a contention window of size $w_i$ is used consecutively $(1/2)\lg(w_i)$ times, and then $w_{i+1} = 2w_i$ for $i\geq 0$. An analog to Lemma~\ref{lem:LLB-star} implies that a lower bound on the number of collisions experienced by LB* provides a lower bound on the number of collisions for LB; we omit this, although we give the important details of the remaining lower-bound argument below.

\begin{theorem}\label{thm:LB-claim}
For a single batch of $n$ packets, with high probability, $\mathcal{C}_{\mbox{\tiny LB}} = \Omega\left(\frac{n\lg n}{\lg\lg n}\right)$.  
\end{theorem}
\begin{proof}
The  argument is similar to that of LLB. Our focus is on a window of size $cn/\lg\lg n$, 
and we state the three main steps of the argument here for completeness:\smallskip

\noindent{\bf (1)} First, the analog of Lemma~\ref{lem:LLB-succ} holds:\smallskip

\noindent {\it For any constant $\epsilon\leq 1$, consider $\epsilon n$ packets and a CW of size  $cn/\lg\lg n$ for a sufficiently small constant $c>0$. With high probability, $O(n/(\lg n)^{d})$ packets succeed in the CW for a constant $d>1$ depending on $\epsilon$ and $c$.}\smallskip

This is proved in the same way, where the number of iterated logarithms is reduced by one in Equations~\ref{eqn:llb1}-\ref{eqn:llb2}, and Corollary~\ref{cor:whp} applies again to give the w.h.p. result. \smallskip

\noindent{\bf (2)} Second, the analog of Lemma~\ref{lem:less-n} holds:\smallskip

\noindent{\it Consider a single batch of $n$ packets executing ~\LB. With high probability, $o(n)$ packets succeed prior to the first window of size $cn/\lg\lg n$ for a sufficiently small  constant $c>0$.}\smallskip

The argument remains essentially unchanged. Equations~\ref{eqn:less-n} to~\ref{eqn:llb3} in the proof of Lemma~\ref{lem:less-n} hold. Next, we ask: how many intervening unique-sized windows exist before reaching the first window of size $cn/\lg\lg n$? This is given by solving for $i$ in the following:
$$\frac{2^i n}{4\ln n} \leq \frac{cn}{\lg\lg n}$$
which yields (again) $i\leq k\lg\lg n$ for some positive constant $k$.

Pessimistically assume each such intervening window has size $cn/\lg\lg n$ (this can only reduce the number of collisions). By the above argument, we have $\epsilon n$ packets for some constant $\epsilon>0$, and for $c$ is sufficiently small, the equivalent of Lemma~\ref{lem:LLB-succ} above guarantees w.h.p. that each window results in  $O(n/(\lg n)^{d}\lg\lg n)$ successful packets for some $d>1$. Each such intervening window executes $O(\lg n)$ times. Therefore, the total number of packets finished over these windows is $O\left(\frac{n}{(\lg n)^{d-1}\lg\lg n}\right) = o(n)$.\smallskip

\noindent{\bf (3)} Third, the equivalent analysis in for deriving Equation~\ref{eqn:llb4} to Equation-\ref{eqn:llb5} holds, showing that the probability of a collision in this window is constant:
\begin{eqnarray}
& = & 1 - \sum_{k=0}^1 \binom{\epsilon n}{k}\hspace{0pt}\left(\frac{\lg\lg n}{cn}\right)^k \hspace{-3pt} \left(1- \frac{\lg\lg n}{cn}\right)^{\epsilon n-k}\\
& = & \Omega(1)
\end{eqnarray}

\noindent \noindent By Corollary~\ref{cor:whp}, the expected number of collisions is tightly bounded:
\begin{eqnarray*}
Pr\left( Z < E[Y] -  t \right)  & = & e^{-\frac{2t^2 \lg\lg n}{cn} }
\end{eqnarray*}
and letting $t = \sqrt{n\ln n}$ is sufficient to prove that the number of collisions in the window is $\Omega\left(\frac{n}{\lg\lg n}\right)$ with high probability.

By specification of LB, the window $w$ is executed $\lg(cn/\lg\lg n) = \Omega(\lg n)$ times. By Step 2, $O(n/(\lg n)^{d}$ $ \lg\lg n)= o(n)$ packets are succeeding in each such execution for some constant $d>1$ and for constant $c>0$ sufficiently small. Thus, there will be at least $\epsilon' n$ packets remaining in each execution of this window size for some sufficiently small constant $\epsilon' > 0$. This means we can apply the above lower bound on the number of collisions per window of size $cn/\lg\lg n$, and so w.h.p. there are a total of $\Omega(n/\lg\lg n) \cdot \Omega(\lg n)$ collisions.\qed
\end{proof}

For completeness, we state the following result, which follows from Theorem~\ref{thm:LB-claim} and the upper bound of $\mathcal{W}_{\mbox{\tiny LB}} = O\left(\frac{n\lg n}{\lg\lg n}\right)$ slots in~\cite{BenderFaHe05}.

\begin{corollary}\label{claim:LB-final}
For a single batch of $n$ packets, with high probability, $\mathcal{C}_{\mbox{\tiny LB}} = \Theta\left(\frac{n\lg n}{\lg\lg n}\right)$. 
\end{corollary}

\vspace{-10pt}


\begin{table*}[t] \vspace{10pt}

\begin{center}

{

\begin{tabular}{  |M{2.2cm}|M{4cm}|M{4cm}|M{4cm}|  }

\hline

\rowcolor{LightCyan} {\bf Algorithm $A$} &  {\boldmath{$\mathcal{W}_A$}}  & {\boldmath{$\mathcal{C}_A$}}  &    {\boldmath{{\bf $\mathcal{T}_A$}}} \\
\hline

{\bf BEB}  &  $\Theta(n\log n)$   & $\Theta(n)$    &   $\Theta\left(n\cdot D + n\log n\right)$ \\  

{\bf LLB}  &  $\Theta\left( \frac{n\log\log n}{\log\log\log n} \right)$   &   $\Theta\left( \frac{n\log\log n}{\log\log\log n} \right)$  & $\Theta\left( \frac{n\log\log n}{\log\log\log n}\cdot D \right)$  \\  

{\bf LB} &   $\Theta\left( \frac{n\log n}{\log\log n} \right)$  &   $\Theta\left( \frac{n\log n}{\log\log n} \right)$  & $ \Theta\left(\frac{n\log n}{\log\log n}\cdot D\right)$   \\  

{\bf STB} &  $\Theta\left(n\right)$   &  $\Theta\left(n\right)$    &  $\Theta\left(n\cdot D\right)$  \\  
\hline

\end{tabular}

}\vspace{-0pt}\caption{New asymptotic bounds on $\mathcal{C}_A$ and $\mathcal{T}_A$. The bounds for $\mathcal{W}_A$ have been previously established~\cite{BenderFaHe05,Gereb-GrausT92,GreenbergL85}, and the bounds for $\mathcal{C}_A$ are presented in this work.}\label{table:new-model}\vspace{-15pt}

\end{center}

\end{table*}


\subsection{Upper Bounding Collisions in STB}\label{sec:stb-analysis}\vspace{-5pt}

Recall from Section~\ref{sec:single-batch} that STB is non-monotonic and executes over a doubly-nested loop. The outer loop sets the current window size $w$ to be double that used in the preceding outer loop. For each such window in the outer loop, the inner loop executes over $\lg w$ windows of decreasing size: $w, w/2, w/4, ..., 1$, and we refer to this sequence of windows as a \defn{run}.

\begin{theorem}
For a single batch of $n$ packets, with high probability, $\mathcal{C}_{\mbox{\tiny STB}}=\Omega(n)$. 
\end{theorem}
\begin{proof}
Consider a run that starts with a window of size $n/8$. The total number of slots up to the end of this window (including  all corresponding runs) is less than $n/2$; therefore, more than $n/2$ packets have not finished by this point. 

In the window $W$ of size $n/4$ that starts the next run, we can calculate the probability of a collision in a slot $j$. Let the indicator random variable $Z_j=1$ if slot $j$ in this window has a collision; $Z_j=0$ otherwise. 
\begin{eqnarray*}
Pr(Z_j=1) &\geq& \sum_{\ell=2}^{n/2} \binom{n}{\ell}\left(\frac{1}{n/4}\right)^{\ell}  \left(1- \frac{1}{n/4}\right)^{(n/2)-\ell}\\
& \geq & \sum_{\ell=2}^{n/2}\left( \frac{n}{\ell}\right)^{\ell} \left(\frac{4}{n}\right)^{\ell}e^{-\frac{(4/n)((n/2) - \ell)}{1-4/n} }\\
&  =  & \sum_{\ell=2}^{n/2}\left( \frac{n}{\ell}\right)^{\ell} \left(\frac{4}{n}\right)^{\ell}e^{-(8/n)((n/2) - \ell) }\\
&  \geq  & \sum_{\ell=2}^{n/2}\left( \frac{n}{\ell}\right)^{\ell} \left(\frac{4}{n}\right)^{\ell} \left(\frac{e^{8\ell/n}}{e^4}\right) \\
&  \geq  & \sum_{\ell=2}^{n/2}\left( \frac{4}{\ell}\right)^{\ell}  \left(\frac{1}{e^5}\right) \\
& = & \Omega(1)
\end{eqnarray*}

\noindent where the second line follows from Fact~\ref{fact:taylor-lower} and a lower bound on binomial coefficient. Therefore, by linearity of expectation, over the $n/4$ bins of the window $W$,  the expected number of collisions  is $\Omega(n)$.

In order to obtain a bound with high probability,  we use Corollary~\ref{cor:whp}.  Letting $Z$ denote the number of collisions in $w$, then:
$$Pr\left( Z < E[Z] - t\right)  \leq  e^{-2t^2/(n/4)}.$$
\noindent Letting $t = \sqrt{dn(\ln n)/8}$ for any constant $d\geq 1$ implies that:
$$Pr\left( Z < E[Z] - t\right)  \leq  e^{-d\ln n} = 1/n^d$$
\noindent and yields the result with high probability.\qed
\end{proof}

Combining the above argument with the known upper bound of $\mathcal{W}_{\mbox{\tiny STB}} = O(n)$~\cite{BenderFaHe05}, the following result follows:
\begin{corollary}\label{claim:LB-final}
For a single batch of $n$ packets, with high probability, $\mathcal{C}_{\mbox{\tiny STB}} = \Theta(n)$. 
\end{corollary}
\vspace{-15pt}


\begin{figure*}[t]\vspace{-1cm}
\captionsetup[subfigure]{labelformat=empty}
\centering
\begin{subfigure}{1.0\textwidth} 
\vspace{30pt}
\includegraphics[width=1.0\textwidth]{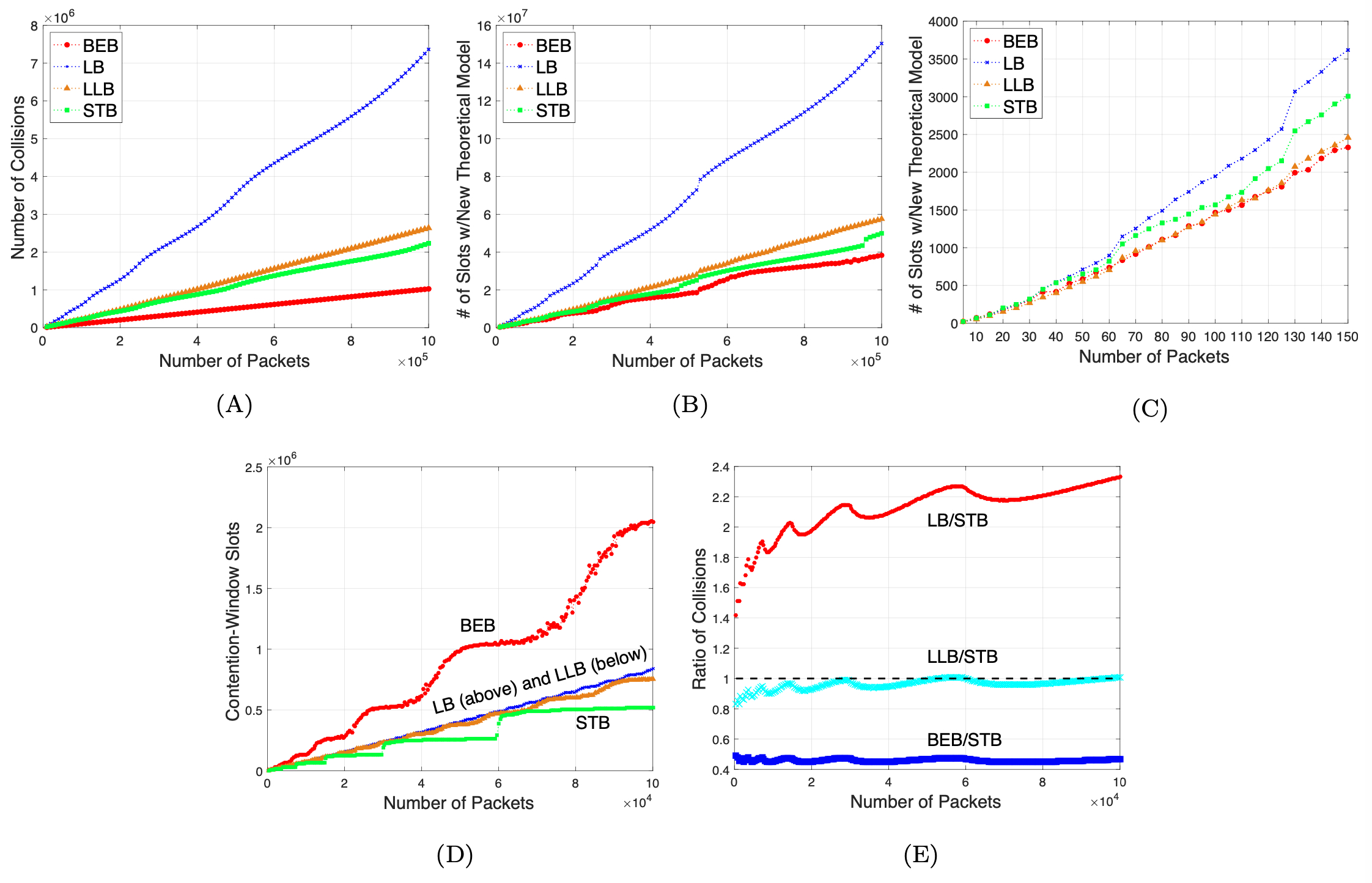} 
\end{subfigure}
 \caption{Plots (A) and (B) depict the number of collisions and slots using our new theoretical model, respectively. The data is generated with our \theoreticalSimulator using $A2^*$ for up to $10^6$  packets, with increments of $10^4$ and $50$ trials per increment. Plot (C) depicts the same simulation, but for the number of packets ranging between $5$ and $150$, with increments of $5$ and $50$ trials per increment. Plots (D) and (E) show results from our \theoreticalSimulator using the original $A2$, with $200$ trials per $n\leq 10^5$, in increments of $400$: (A) CW slots with median values plotted, (B) ratio of median number of collisions for BEB, LLB, and LB versus STB.  Confidence intervals are omitted, since they are negligible.} \label{fig:java-sims-large}
\vspace{5pt}
\end{figure*}

\subsection{Applying Our New Theoretical Model}\label{sec:as-total}

Plugging our results for $\mathcal{C}_{A}$ into our new theoretical model from Section~\ref{sec:theory}, we derive bounds for $\mathcal{T}_A$, and  these are presented in Table~\ref{table:new-model}.  Clearly, the performance of these backoff algorithms depends on $D$. For example, if $D=O(1)$, then LB, LLB, and STB remains asymptotically superior to BEB. On the other hand, this is not true for sufficiently large $D$. Thus,  in our theoretical model, $D$ dictates the mix between traditional makespan, as measured by CW slots, and the cost of collisions.

We believe that this model may  apply generally to network scenarios where $D$ is a significant parameter. Here,  for our WiFi setting, we believe that $D$ is large enough to explain the discrepancy observed in our experiments, and we further speculate below that $D=\Omega(\log n)$.\medskip\smallskip

\noindent{\bf Scaling of {\boldmath{$D$}}.} In WiFi networks,  $D$ can capture the time corresponding to the preamble, ACK timeouts, and transmission delay, as discussed in Section~\ref{sec:hidden}.  The preamble remains fixed for our setting. In contrast, both ACK timeouts and transmission delay do scale with network size. The former scales slowly; a common heuristic is to increase the timeout by $1\mu s$ for every $300m$-increase in distance between the sender and  receiver.  

How does transmission delay scale? If we assume it increases with packet size, then this can be broken into two parts: the amount of payload, and the amount of control information (for example, the header of a packet). Focusing on the control information, note that the packet header must contain enough bits to uniquely address all $n$ stations. Therefore, as $n$ grows, the number of bits required to address devices must also increase, and we may assume packet size, and thus transmission delay, scales as the logarithm of $n$.  

Given this discussion, for large values of $n$, one may argue that $D$ ought to be treated as a slowly growing function of $n$, such as $\Omega(\log n)$.

\medskip\smallskip

\noindent{\bf Implications.} For $D=\Omega(\log n)$, both $\mathcal{T}_{LB}$ and $\mathcal{T}_{LLB}$ exceed $\mathcal{T}_{BEB}$ asymptotically.

\medskip
\begin{tcolorbox}[standard jigsaw, opacityback=0]
\noindent{\bf Result 5.~} {\it Bounds derived under our new theoretical model imply that LLB and LB \hspace{-1pt}should \hspace{-1pt}underperform \hspace{-1pt}BEB \hspace{-1pt}for \hspace{-1pt}sufficiently \hspace{-1pt}large \hspace{-1pt}$n$ \hspace{-1pt}and \hspace{-1pt}$D$.} \end{tcolorbox}

Can we verify this via simulation? To see the scaling behavior, we execute these algorithms using our \theoreticalSimulator using our new assumption $A2^*$\hspace{-3pt}, with $D=\log_2(k)$, where $k$ is the current batch size being simulated. We explore up to $n\leq 1,000,000$, with batch sizes that increase by $10,000$.  

Figure~\ref{fig:java-sims-large}(A) illustrates the number of collisions. We witness the relative behavior predicted by our new theoretical bounds on $\mathcal{C}_{LLB}$ and $\mathcal{C}_{LB}$.  In particular, LLB and (especially) LB have many more collisions than BEB. Using our new theoretical model, the number of slots is plotted in Figure~\ref{fig:java-sims-large}(B); again, we see the relative performance predicted by our theoretical bounds on $\mathcal{T}_{LLB}$ and $\mathcal{T}_{LB}$; that is, LLB and LB underperform BEB.  We also note that this looks significantly different than the behavior given by the old theoretical model; this is illustrated by contrasting with Figure~\ref{fig:java-sims-large}(D), which depicts the number of CW slots (i.e., $\mathcal{W}_{A}$ under the original theoretical model).

Figure~\ref{fig:java-sims-large}(D) also resolves an issue noted in Section~\ref{sec:CW-slots-exp}, where LLB was seen to incur more CW slots than LB, despite the former being theoretically superior along this metric. Results from the \theoreticalSimulator using the original $A2$ show that indeed LLB outperforms LB for sufficiently large $n$.

What about at small scale? We do not necessarily expect our model to match the results of the \detailedSimulator, but it is worth examining. Scoping into the range of $n\leq 150$, we observe in Figure \ref{fig:java-sims-large}(C) an ordering that is close to---but not exactly a match with---Figures~\ref{fig:total-time}(A) and (B). In particular, the only difference is the ordering of LB and STB. However, we note that LB and STB are very similar under the \detailedSimulator, and they remain close here also under the \theoreticalSimulator using $A2^*$. We suspect that a constant factor in the analysis may explain the difference at this small scale.

\medskip\smallskip

\noindent{\bf Payload Size.} In the above discussion, we speculated about the asymptotics of these backoff algorithms via the amount of control information as a function of $n$. However, payload size may also be sufficient to cause significant deviation from the original A2, and our model captures this also. We observed in Section~\ref{sec:totaltime} that increasing the packet size favored BEB over LLB, the latter being the closest competitor to BEB in our experiments. 

To investigate this further,  we use the \detailedSimulator to examine the relative performance of these two algorithms as packet size increases; this is illustrated in Figure~\ref{fig:packet-size-increase}. As the packet size grows, LLB performs increasingly worse than BEB. We fit a linear regression model of the value LLB - BEB on the number of packets. This fitting model implies that when the payload size  increases by $100$B, the average increase in \totaltime for LLB is roughly $700 \mu s$ more than the increase experienced by BEB. The increase rate is statistically significant ($p$-value less than $0.001$). 

\begin{figure}[t]
\centering
\includegraphics[width=0.5\textwidth, trim = 0.1in 0in 1in 2.1in, clip, scale=0.6]{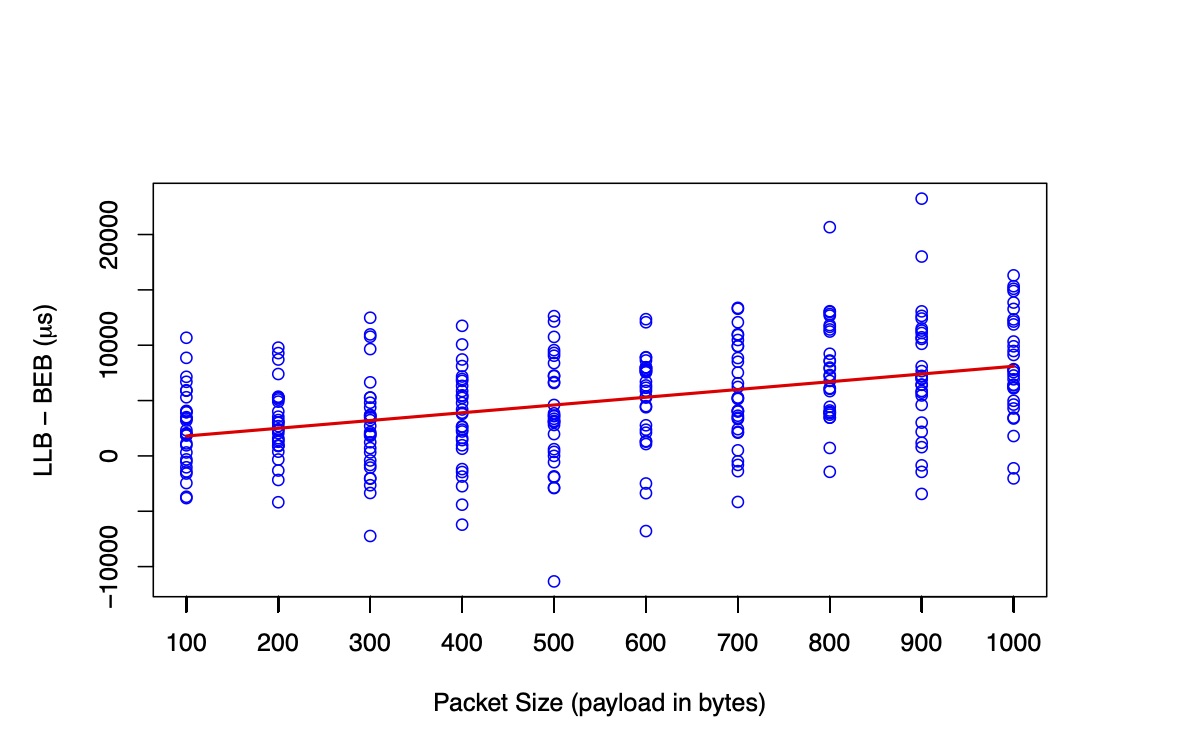}
\caption{The difference in \totaltime between LLB and BEB for $n=150$ as packet size increases, with 30 trials per packet size.}\label{fig:packet-size-increase}
\end{figure}

\smallskip

\noindent{\bf Performance of STB.} At this point, the elephant in the room is STB. For $D=\Omega(\log n)$, we observe that   $\mathcal{T}_{BEB} = \Theta(\mathcal{T}_{STB})$, and so we should expect these algorithms to behave similarly, up to a constant factor.  We also note that, by our analysis of $\mathcal{T}_{A}$, STB should be outperforming LLB and LB.

Yet, in our experiments with the \detailedSimulator, STB has the worst \totaltime for $n=10$ to $150$. However, we claim that this is only occurring at small scale, and we offer the following explanation of this behavior.

Observe that, empirically, STB has roughly twice as many collisions as BEB. This shows up experimentally in our Figure~\ref{fig:java-sims-large}(A). But we can show this more clearly in Figure~\ref{fig:java-sims-large}(E), where we observe that the number of collisions for STB is larger than BEB by roughly a factor of $2$ over this large range of $n$. Note that the plot of BEB/STB is (roughly) flat, as expected from our asymptotic analysis of collisions. 

This implies that two things should be true for large $n$. First, STB should approximate the performance of BEB and eventually outcompete LLB and LB. Second, STB should always be doing worse than BEB given the ratio of collisions. Indeed, we observe both of these behaviors to be true in Figure~\ref{fig:java-sims-large}(B). Therefore, it is possible for STB to do poorly by comparison to LLB and LB at small scale, but to outcompete for larger values of $n$, which  aligns with the performance predicted by our bounds on $\mathcal{T}_{\mbox{\tiny STB}}$. 


\begin{figure*}[t] 
\captionsetup[subfigure]{labelformat=empty}
\centering
\begin{subfigure}{1.0\textwidth} 
\includegraphics[width=1.0\textwidth]{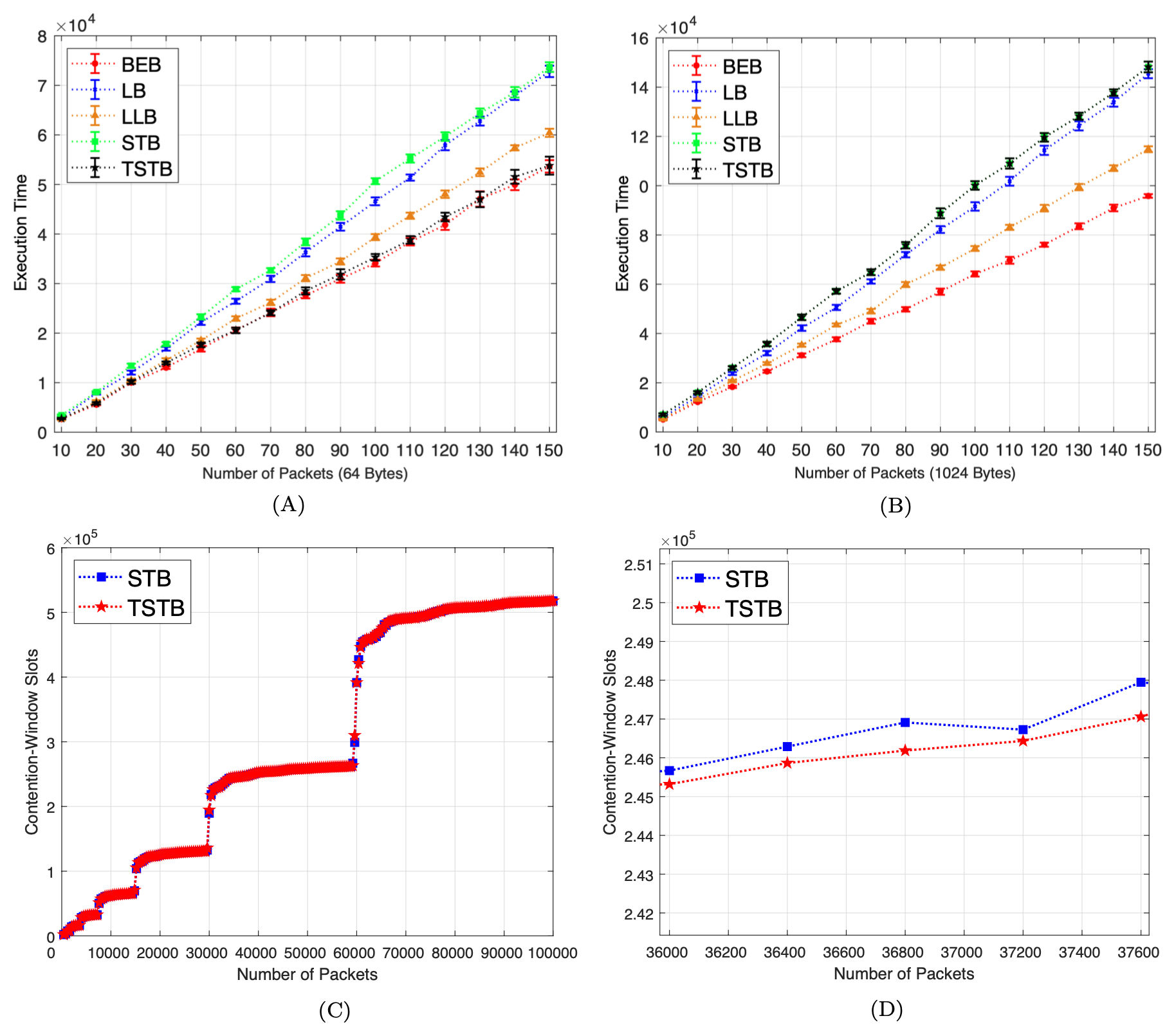} 
\end{subfigure}
\caption{Median values are reported: (A)  \Totaltime from the \detailedSimulator with $30$ trials for each value of $n$ with $64$B payload with $2m$ distance, and TSTB truncated after the first window; (B) \Totaltime  from the \detailedSimulator  with $30$ trials for each value of $n$ with $1024$B payload with $2m$ distance, and TSTB truncated after the fifth window; (C) CW slots via \theoreticalSimulator with $100$ trials for each value of $n$ ranging from $400$ to $100000$; (D) A zoomed-in portion of plot (C) illustrating the minor improvement of TSTB over STB in terms of CW slots. Bars represent $95\%$ confidence intervals, and they are omitted in (C) and (D), since they are negligible.} \label{fig:makespan-tstb}
\end{figure*} 
 
\vspace{-10pt}
\section{Discussion}\label{sec:discussion} \vspace{-2pt}

In this section, we return to the topic of TSTB's performance, which we deferred in Section~\ref{sec:single-batch}. Then, we discuss the applicability of our findings in the context of other protocols and networks. \vspace{-10pt}


\subsection{Performance of TSTB}\label{sec:tstb}

Recall from Section~\ref{sec:single-batch} that TSTB proceeds like STB, except that its runs are truncated, and the degree of this truncation depends on a constant $c>0$. We implemented TSTB in order to compare its performance against the other algorithms.

To begin our discussion, we consider two extreme behaviors that depend on the value $c$. When $c$ is sufficiently small, the run is truncated after the first window, and TSTB behaves identically to BEB. Conversely, when $c$ is sufficiently large, the run is not truncated at all, and TSTB behaves identically to STB. These two cases are illustrated in Figures~\ref{fig:makespan-tstb} (A) and (B), respectively. Given this view of TSTB, it is not surprising that it cannot outcompete BEB in terms of \totaltime over the small range of $n$-values investigated in our experiments with the \detailedSimulator. We explored a variety of values for $c$, and none gave an improvement.

Of course, in terms of CW slots, when TSTB  behaves closely to STB, it will outperform BEB. Indeed, over large values of $n$, the number of CW slots used by TSTB is smaller, although not by a significant amount. This behavior is illustrated in  Figures~\ref{fig:makespan-tstb} (C) and (D), with the latter demonstrating the minor improvement by TSTB. While we do not include this analysis, it is easy to see that TSTB will also incur $\Theta(n)$ collisions. 

These observations---especially those regarding \totaltime---are the reason for postponing discussion of TSTB until this point.


\subsection{Scope of Our Findings}\label{sec:interpret} 

In this section, we consider to what extent our findings are an artifact of IEEE 802.11g, and whether LB, LLB, STB might do better inside other wireless protocols.\smallskip

\noindent{\bf IEEE 802.11g uses a truncated BEB, is this significant?} In our experiments, the maximum congestion-window size is $4096$ which differs from the abstract model where no such upper bound exists. However, even for $n=150$, this maximum is never reached during an execution of BEB and this does not seem to have any noticeable impact on the trend observed in Figures~4(A) and 4(B). 
\smallskip\smallskip

\noindent{\bf What if smaller packets are used?} During a collision, the time lost to transmitting would be reduced. In an extreme case, if the transmission of a packet fit within a slot, this would align more closely with A2.

Due to overhead, packet size has a lower bound in IEEE 802.11. Additionally, in NS3, there is a 12-byte payload minimum which translates into a minimum total packet size of $76$ bytes for our experiments.\footnote{This is set within the \texttt{UdpClient} class of NS3.} For this packet size, we witnessed the same qualitative behavior is observed in terms of CW slots and \totaltime.  

Alternatives to 802.11 might see more significant decreases. However, there is a tradeoff for any protocol. A smaller packet implies a reduced payload given the need for control information (for routing, error-detection, etc.) and this means that  throughput is degraded.\smallskip\smallskip

\noindent{\bf What if the ACK-timeout duration is reduced or acknowledgements are removed altogether?} This would also bring us closer to A2. In our experiments, the ACK-timeout is $75\mu$s (recall Section~\ref{sec:experimental}) and values significantly below this threshold will lead a station to consider its packet lost before the ACK can be received. This results in unnecessary retransmissions and, ultimately, poor throughput.

Totally removing acknowledgements (or some form of feedback) is difficult in many settings since, arguably, they are critical to {\it any} protocol that provides reliability; more so when transmissions are subject to disruption by other stations over a shared channel. \smallskip\smallskip

\noindent{\bf Would using TCP rather than UDP significantly alter these findings?} We conjecture that using TCP may strengthen the effects we observe. A collision will register as a lost packet, and this can trigger congestion control, which would result in additional delays.  Additionally, packet headers are larger with TCP than UDP, which should increase the transmission delay. Both of these could increase the value of $D$ in our new theoretical model. That said, this is a prediction, and an empirical investigation would be needed to determine the degree to which our new theoretical model applies.\medskip\smallskip


\noindent{\bf To what extent do these findings generalize to other protocols?} We do {\it not} claim that our findings hold for all protocols. If (a) sufficiently small packets are feasible {\it and} (b) reliability is not paramount, performance should align better with theoretical guarantees derived from using assumption A2. 

We {\it do} claim that our findings should apply to several other protocols. Examples include members of the IEEE 802.11 family,  IEEE 802.15.4, and IEEE 802.16 (WiMAX); these employ some form of backoff  for data transmission, incur overhead from control information, and use feedback via acknowledgements or a timeout to determine success or failure.  This is a significant slice of current wireless standards.  

Bluetooth (IEEE 802.15.1) is another popular standard where we might also expect our findings to apply.  Here, devices form piconets using a master/slave architecture, and the master coordinates the channels used for communication by its slave devices; thus collisions typically arise between different, colocated piconets. In contrast to WiFi, IEEE 802.15.4, and WiMAX, Bluetooth has limited dependence on backoff, where it is used by a device to join a piconet and pair with a master (see Section 8.10.7.2 of the specification~\cite{802.15.1-standard}). Furthermore, this backoff procedure is basic, using a single window of size $1024$ during the pairing procedure, regardless of the number of devices. 

That said, with a peak data rate in the low Mbps~\cite{beard:wireless}, even sending $128$ bytes will incur roughly $10^3 \mu s$.  Additionally, acknowledgements are used, with Bluetooth employing a standard automatic repeat request (ARQ) process, along with forward error correction. For these reasons, (a) and (b) do not seem to apply, and we should expect a cost incurred due to collisions, which could potentially be analyzed under our new model.

\smallskip

\begin{tcolorbox}[standard jigsaw, opacityback=0]
\noindent{\bf Result 6.} {\it Designing contention-resolution algorithms using assumption A2 seems likely to translate into poor performance in practice for a range of wireless protocols.}\end{tcolorbox}

A setting where the  original  theoretical model may be valid is networks of multi-antenna devices. If a collision can be detected more efficiently, perhaps by a separate antenna,  the delay can be reduced. Canceling the signal at the sending device so that other transmissions (that would cause the collision) can be detected is challenging. However, this  is possible (for an interesting application, see~\cite{Gollakota:2011:THY:2018436.2018438}) and such schemes have been proposed using multiple-input multiple-output (MIMO) antenna technology~\cite{6963622,6962145}.

Finally, we note that future standards may satisfy (a) and (b). A possible setting is the Internet-of-Things (IoT);  for example,~\cite{7524360} characterizes  IoT transmissions as ``small'' and ``intermittent, delay-sensitive, and short-lived". To reduce delay, the authors argue for removing much of the control messaging used by traditional MAC protocols. Therefore, this setting seems more closely aligned with A2. However, using this same logic,~\cite{7524360} also argues for the removal of any backoff-like contention-resolution mechanism. So, these standards  are  in flux and we may indeed see protocols that avoid the issues we identify here. 

 
\section{Related Work}\label{sec:related}

A preliminary version of this work appeared~\cite{anderton:is}. Here, we have revised and expanded on much of the material. Experiments regarding CW slots and \totaltime for all algorithms  have been redone using different distances between nodes in order to verify that the results hold. The superior performance of BEB with respect to \totaltime is explored further via additional  experiments measuring the number of collisions (Section~\ref{sec:hidden}). We propose a new theoretical model to capture the cost of collisions and our analytical arguments are presented in full (Section~\ref{sec:theory}). Finally, we also investigated a new algorithm, TSTB, in order to measure its performance (Section~\ref{sec:tstb}). \medskip


\noindent{\bf Theoretical Results.} There is a vast body of theoretical literature on the problem of contention resolution. Exponential backoff has been studied under the case where arrival times of packets are stochastic~(see~\cite{GoldbergMa96a,GoldbergMaPaSr00,HastadLeRo87,RaghavanUp95,Capetanakis:2006}). Guarantees on stability are known~\cite{Al-Ammal2000,Al-Ammal2001,Goodman:1988:SBE:44483.44488}, and under saturated conditions~\cite{bianchi:performance}. 

Moving away from stochastic arrival times, batched arrivals have received significant attention~\cite{BenderFaHe05,zhou2019singletons}. Results on makespan  are also known for batched arrivals when packets can have different sizes~\cite{bender:heterogeneous}. In this case, BEB is again found to offer poor makespan, although the model deviates significantly from wireless communication protocols. In particular, large packets being transmitted can be interrupted at any point by other stations that transmit on the channel without performing carrier-sense multiple access.

Energy efficiency is important to multiple access in many low-power wireless networks~\cite{jurdzinski:energy,chang:exponential,bender:contention,BenderKPY18}.  The communication channel may be subject to adversarial disruption (for examples, see~\cite{king:conflict,DBLP:journals/dc/KingPSY18,richa:competitive-j,gilbert:near,Bender:2015:RA:2818936.2818949,DBLP:journals/topc/GilbertZ15,DBLP:conf/podc/ChenZ20,DBLP:conf/spaa/ChenZ19,DBLP:conf/spaa/GilbertZ13}), several results address the challenge of multiple access~\cite{awerbuch:jamming,richa:jamming2,richa:jamming3,richa:jamming4,bender:how,ogierman:competitive,richa:efficient-j,aldawsari:adversarial,Anantharamu2011,DBLP:journals/jcss/AnantharamuCKR19,Tan2014,bender:how,BenderFGY19}.

Regarding the time required for a single successful transmission, a lower-bound of $\Omega(\log\log n)$ is known~\cite{willard:loglog}. Recent work has shown how predictions regarding the number of packets can improve performance~\cite{gilbert:contention}. Additionally, different channel-feedback models have been investigated~\cite{fineman:contention2,fineman:contention,bender2020contention}.  The wake-up problem is closely-related, since it addresses how long it takes for a collection of  devices to receive a transmission that wakes up all devices~\cite{chlebus:better,chlebus:wakeup,chrobak:wakeup,Chlebus:2016:SWM:2882263.2882514,Jurdzinski:2015:CSM:2767386.2767439}. 

Deterministic contention-resolution algorithms have been considered~\cite{DBLP:conf/icdcs/MarcoKS19,DBLP:conf/podc/Kowalski05}. Additionally, deterministic  protocols for the closely-related problem of broadcast on a multiple-access channel have also received significant attention~\cite{ChlebusKoRo06,ChlebusKoRo12,anantharamu:adversarial-opodis}. In contrast, the performance of adaptive algorithms---where packets make use of channel feedback---and the power of randomized versus deterministic approaches has been examined. When players have access to a global clock deterministic, non-adaptive contention-resolution has been investigated~\cite{doi:10.1137/140982763}. Conversely, without a global clock, both adaptive and non-adaptive results are known~\cite{DeMarco:2017:ASC:3087801.3087831}. 

Several results incorporate issues of scheduling or time-constrained access to a shared resource~\cite{bienkowski:distributed,anta:scheduling,bienkowski2016randomized,DBLP:conf/spaa/AgrawalBFGY20}. 

Finally, a number of approaches attempt to estimate the the number of devices contending for the channel in order to start with a larger initial CW~\cite{hajek:decentralized,kelly:decentralized,Gerla:1977,cali:dynamic,cali:design,bianchi:kalman,bender:contention,BenderKPY18}. Preliminary results suggest that such size-estimation algorithms may be promising~\cite{anderton:is}. We omit this content here given that the results presented are already substantial, and the performance of size-estimation results deserve an in-depth investigation that we believe falls outside the scope of the current work.

\medskip

\noindent{\bf Performance under 802.11.} In terms of the performance of backoff algorithms in practice, the performance of IEEE 802.11 has received significant attention~\cite{Ni:survey,Kuptsov201437,5191039,1543687,LiHSC03,duda:understanding,5963276}. There are several results that focus specifically on the performance of BEB within IEEE 802.11, and we summarize those that are most closely related.

Under continuous traffic, windowed backoff schemes are examined in~\cite{sun:backoff} with a focus on the tradeoff between throughput and fairness. The authors focus on polynomial backoff and demonstrate via analysis and NS2 (the predecessor to NS3) simulations that quadratic backoff is a good candidate with respect to both metrics.

Work in~\cite{1424043} addresses saturated throughput (each station always has a packet ready to be transmitted) of exponential backoff; roughly, this is the maximum throughput under stable packet arrival rates. Custom  simulations are used to confirm these findings. 

In~\cite{6859627}, the authors propose backoff algorithms where the size of the contention window is modified by a small constant factor based on the number of successful transmissions observed. NS2  simulations are used to demonstrate improvements over BEB within 802.11 for a steady stream of packets (i.e. non-bursty traffic). 

Lastly, in~\cite{saher:log}, the authors examine a variation on backoff where the contention window increases multiplicatively by the logarithm of the current window size (confusingly, also referred to as ``logarithmic backoff''). NS2  simulations imply an advantage to their variant over BEB within IEEE 802.11, again for non-bursty traffic.

\section{Concluding Remarks and Open Questions}

We have offered evidence that a commonly-used model for designing contention-resolution algorithms is not adequately accounting for the cost of collisions in the domain of WiFi and related wireless communication standards. Interesting questions remain.  Experimentally, we would like to know for certain   whether these findings hold for other popular protocols, such as Bluetooth, or under different transport layer protocols, such as TCP. We hypothesize that the results should be similar, but experimental confirmation would be useful. In terms of analytical work, we have argued for why collisions have an impact on performance, but what is the optimal tradeoff between collisions and CW slots? Assuming that this tradeoff is known, can we design algorithms that leverage this information? Can we analyze the more general case where packets do not arrive in batches, but may instead  arrive at arbitrary times?\\

\noindent{\bf Acknowledgements.} We are grateful to the reviewers for their comments, which greatly improved our manuscript.





\end{document}